\documentclass[reqno]{amsart}
\usepackage[margin=1.5in,bottom=1.25in]{geometry}		

\usepackage{amsmath}		
\usepackage{amssymb}		
\usepackage{amsfonts}		
\usepackage{amsthm}		
\usepackage[foot]{amsaddr}		

\usepackage{mathtools}		
\usepackage{tabularx}

\mathtoolsset{%
}

\usepackage[utf8]{inputenc}		
\usepackage[T1]{fontenc}		

\usepackage[
cal=cm,
]
{mathalfa}

\usepackage{dsfont}		

\usepackage[sf,mono=false]{libertine}

\usepackage{acronym}		
\renewcommand*{\aclabelfont}[1]{\acsfont{#1}}		
\newcommand{\acdef}[1]{\emph{\acl{#1}} \textup{(\acs{#1})}\acused{#1}}		

\usepackage[labelfont={bf,small},labelsep=colon,font=small]{caption}	
\usepackage[dvipsnames,svgnames]{xcolor}		
\colorlet{MyRed}{Crimson!75!Black}
\colorlet{MyGreen}{DarkGreen!80!Black}
\colorlet{MyBlue}{MediumBlue}

\newcommand{\afterhead}{}		
\newcommand{\para}[1]{\subsubsection*{\textbf{#1}\afterhead}}

\usepackage{latexsym}		

\usepackage{subcaption}		
\usepackage{tikz}		
\usetikzlibrary{calc,patterns}		

\usepackage{array}		
\usepackage{booktabs}		
\usepackage[inline,shortlabels]{enumitem}		
\setenumerate{itemsep=\smallskipamount,topsep=\smallskipamount}		

\usepackage[kerning=true]{microtype}		

\usepackage{xspace}		

\usepackage{cite}
\newcommand{\citet}{\cite}
\newcommand{\citep}{\cite}

\usepackage{hyperref}
\hypersetup{
colorlinks=true,
linktocpage=true,
pdfstartview=FitH,
breaklinks=true,
pdfpagemode=UseNone,
pageanchor=true,
pdfpagemode=UseOutlines,
plainpages=false,
bookmarksnumbered,
bookmarksopen=false,
bookmarksopenlevel=1,
hypertexnames=true,
pdfhighlight=/O,
urlcolor=MyBlue,linkcolor=MyBlue,citecolor=MyBlue,	
pdftitle={},
pdfauthor={},
pdfsubject={},
pdfkeywords={},
pdfcreator={pdfLaTeX},
pdfproducer={LaTeX with hyperref}
}

\newcommand{\EMAIL}[1]{\email{\href{mailto:#1}{#1}}}

\usepackage[sort&compress,capitalize,nameinlink]{cleveref}		
\crefname{assumption}{Assumption}{Assumptions}

\usepackage{algorithm}		
\usepackage{algpseudocode}		

\usepackage{thmtools}		
\usepackage{thm-restate}		

\theoremstyle{plain}

\newtheorem*{corollary*}{Corollary}		

\theoremstyle{definition}
\newtheorem{definition}{Definition}		
\newtheorem{example}{Example}		

\newtheorem*{definition*}{Definition}		
\newtheorem*{assumption*}{Assumptions}		
\newtheorem*{example*}{Example}		

\theoremstyle{remark}

\newtheorem*{remark*}{Remark}		

\def\endenv{\hfill{\small\S}}

\newenvironment{Proof}[1][Proof]{\begin{proof}[#1]}{\end{proof}}

\usepackage[showdeletions]{color-edits}		
\usepackage[normalem]{ulem}		

\newcommand{\debug}[1]{#1}		

\newcommand{\newmacro}[2]{\newcommand{#1}{\debug{#2}}}		
\newcommand{\newop}[2]{\DeclareMathOperator{#1}{\debug{#2}}}		

\DeclarePairedDelimiter{\braces}{\{}{\}}		
\DeclarePairedDelimiter{\bracks}{[}{]}		
\DeclarePairedDelimiter{\parens}{(}{)}		

\DeclarePairedDelimiter{\abs}{\lvert}{\rvert}		
\DeclarePairedDelimiter{\floor}{\lfloor}{\rfloor}		

\DeclarePairedDelimiterX{\setdef}[2]{\{}{\}}{#1:#2}		
\DeclarePairedDelimiterXPP{\exclude}[1]{\mathopen{}\setminus}{\{}{\}}{}{#1}

\newcommand{\N}{\mathbb{N}}		
\newcommand{\Z}{\mathbb{Z}}		
\newcommand{\R}{\mathbb{R}}		

\DeclareMathOperator*{\argmax}{arg\,max}		
\DeclareMathOperator*{\argmin}{arg\,min}		

\DeclareMathOperator{\dist}{dist}		
\DeclareMathOperator{\one}{\mathds{1}}		
\DeclareMathOperator{\relint}{ri}		

\newcommand{\cf}{cf.\xspace}		
\newcommand{\eg}{e.g.,\xspace}		
\newcommand{\ie}{i.e.,\xspace}		

\newcommand{\textpar}[1]{\textup(#1\textup)}		

\newcommand{\alt}[1]{#1'}		
\newcommand{\altalt}[1]{#1''}		

\newmacro{\dd}{\:d}		
\newcommand{\eps}{\varepsilon}		

\newcommand{\insum}{\sum\nolimits}		

\newmacro{\const}{c}		
\newmacro{\coef}{\lambda}		
\newmacro{\param}{\theta}		
\newmacro{\params}{\Theta}		

\newmacro{\pexp}{p}		
\newmacro{\qexp}{q}		
\newmacro{\rexp}{r}		

\newmacro{\beforestart}{0}		
\newmacro{\start}{1}		
\newmacro{\afterstart}{2}		
\newmacro{\running}{\start,\afterstart,\dotsc}		

\newmacro{\run}{n}		
\newmacro{\runalt}{k}		
\newmacro{\runaltalt}{m}		
\newmacro{\nRuns}{T}		
\newmacro{\runs}{\mathcal{\nRuns}}		

\newmacro{\state}{X}		
\newmacro{\statealt}{Y}		
\newmacro{\statealtalt}{Z}		

\newcommand{\init}[1][\state]{\debug{#1}_{\start}}		

\newcommand{\curr}[1][\state]{\debug{#1}_{\run}}		
\renewcommand{\next}[1][\state]{\debug{#1}_{\run+1}}		

\newmacro{\tstart}{0}		
\newmacro{\timealt}{s}		
\newmacro{\horizon}{T}		

\newmacro{\traj}{x}		
\newmacro{\trajalt}{y}		
\newmacro{\trajaltalt}{z}		

\newmacro{\flowmap}{\Theta}		
\DeclarePairedDelimiterXPP{\flowof}[2]{\flowmap_{#1}}{(}{)}{}{#2}		

\newop{\Nash}{NE}		
\newop{\CE}{CE}		
\newop{\CCE}{CCE}		
\newop{\NI}{NI}		

\newop{\brep}{br}		
\newop{\reg}{Reg}		
\newop{\preg}{\overline{Reg}}		
\newop{\val}{val}		

\newmacro{\play}{i}		
\newmacro{\playalt}{j}		
\newmacro{\playaltalt}{k}		
\newmacro{\nPlayers}{N}		
\newmacro{\players}{\mathcal{\nPlayers}}		

\newmacro{\pure}{\alpha}		
\newmacro{\purealt}{\beta}		
\newmacro{\purealtalt}{\gamma}		
\newmacro{\pureq}{\eq[\pure]}		
\newmacro{\nPures}{A}		
\newmacro{\pures}{\mathcal{\nPures}}		

\newmacro{\strat}{x}		
\newmacro{\stratalt}{\alt\strat}		
\newmacro{\strataltalt}{\altalt\strat}		
\newmacro{\strats}{\mathcal{X}}		
\newmacro{\intstrats}{\strats^{\circle}}		

\newmacro{\score}{y}		
\newmacro{\scorealt}{\alt\score}		

\newcommand{\eq}{\sol}		

\newmacro{\loss}{\ell}		
\newmacro{\pay}{u}		
\newmacro{\payv}{v}		
\newmacro{\pot}{f}		

\newmacro{\game}{\mathcal{G}}		
\newmacro{\gameall}{\game(\players,\points,\loss)}		

\newmacro{\fingame}{\Gamma}		
\newmacro{\fingameall}{\Gamma(\players,\pures,\pay)}		

\newmacro{\gmat}{g}		
\newmacro{\gdist}{\dist_{\gmat}}
\newmacro{\mfld}{M}		
\newmacro{\form}{\omega}		

\newmacro{\tvec}{z}		
\newmacro{\uvec}{u}		

\newmacro{\ball}{\mathbb{B}}		
\newmacro{\sphere}{\mathbb{S}}		

\newmacro{\vertex}{v}		
\newmacro{\vertexalt}{\alt\vertex}		
\newmacro{\vertexaltalt}{\altalt\vertex}		
\newmacro{\nVertices}{V}		
\newmacro{\vertices}{\mathcal{\nVertices}}		

\newmacro{\edge}{e}		
\newmacro{\edgealt}{\alt\edge}		
\newmacro{\edgealtalt}{\altalt\edge}		
\newmacro{\nEdges}{E}		
\newmacro{\edges}{\mathcal{\nEdges}}		

\newmacro{\graph}{\mathcal{G}}		
\newmacro{\graphall}{\graph(\vertices,\edges)}		

\newmacro{\vecspace}{\mathcal{V}}		
\newmacro{\subspace}{\mathcal{W}}		

\newmacro{\bvec}{e}		
\newmacro{\bvecs}{\mathcal{E}}		
\newmacro{\vvec}{v}		

\newmacro{\coord}{i}		
\newmacro{\coordalt}{j}		
\newmacro{\coordaltalt}{k}		
\newmacro{\nCoords}{d}		
\newmacro{\dims}{\nCoords}		
\newmacro{\vdim}{\nCoords}		

\newmacro{\pspace}{\mathcal{X}}		
\newmacro{\dspace}{\mathcal{Y}}		

\newmacro{\ppoint}{x}		
\newmacro{\ppointalt}{\alt\ppoint}		
\newmacro{\ppointaltalt}{\altalt\ppoint}		
\newmacro{\ppoints}{\mathcal{X}}		
\newmacro{\pstate}{X}		

\newmacro{\dpoint}{y}		
\newmacro{\dpointalt}{\alt\dpoint}		
\newmacro{\dpointaltalt}{\altalt\dpoint}		
\newmacro{\dpoints}{\mathcal{Y}}		
\newmacro{\dstate}{Y}		

\newmacro{\pvec}{u}		
\newmacro{\dvec}{v}		

\newmacro{\mat}{M}		
\newmacro{\hmat}{H}		

\newmacro{\ones}{\mathbf{1}}		
\newmacro{\eye}{I}		
\newmacro{\zer}{\mathbf{0}}		

\DeclarePairedDelimiter{\norm}{\lVert}{\rVert}		
\DeclarePairedDelimiterXPP{\dnorm}[1]{}{\lVert}{\rVert}{_{\ast}}{#1}		

\DeclarePairedDelimiterXPP{\onenorm}[1]{}{\lVert}{\rVert}{_{1}}{#1}		
\DeclarePairedDelimiterXPP{\twonorm}[1]{}{\lVert}{\rVert}{_{2}}{#1}		
\DeclarePairedDelimiterXPP{\supnorm}[1]{}{\lVert}{\rVert}{_{\infty}}{#1}		

\DeclarePairedDelimiterX{\braket}[2]{\langle}{\rangle}{#1,#2}		

\DeclarePairedDelimiterX{\inner}[2]{\langle}{\rangle}{#1,#2}		

\newcommand{\defeq}{\coloneqq}		

\newcommand{\from}{\colon}		

\newmacro{\source}{O}		
\newmacro{\sink}{D}		

\newmacro{\pair}{i}		
\newmacro{\pairalt}{j}		
\newmacro{\pairaltalt}{k}		
\newmacro{\nPairs}{N}		
\newmacro{\pairs}{\mathcal{\nPairs}}		

\newmacro{\route}{p}		
\newmacro{\routealt}{\alt\route}		
\newmacro{\routealtalt}{\altalt\route}		
\newmacro{\nRoutes}{P}		
\newmacro{\routes}{\mathcal{\nRoutes}}		

\newmacro{\flow}{f}		
\newmacro{\flowalt}{\alt\flow}		
\newmacro{\flowaltalt}{\altalt\flow}		
\newmacro{\flows}{\mathcal{F}}		

\newmacro{\load}{x}		
\newmacro{\loadalt}{\alt\load}		
\newmacro{\loadaltalt}{\altalt\load}		
\newmacro{\loads}{\mathcal{X}}		

\newop{\Opt}{Opt}		
\newop{\Sol}{Sol}		
\newop{\gap}{Gap}		
\newop{\orcl}{Or}		

\newmacro{\obj}{f}		
\newmacro{\objalt}{g}		
\newmacro{\sobj}{F}		

\newmacro{\gvec}{g}		
\newmacro{\oper}{A}		
\newmacro{\vecfield}{v}		

\newcommand{\sol}[1][\point]{#1^{\ast}}		

\newmacro{\vbound}{G}		
\newmacro{\lips}{L}		
\newmacro{\strong}{\alpha}		
\newmacro{\smooth}{\beta}		

\newop{\tspace}{T}		
\newop{\tcone}{TC}		
\newop{\dcone}{\tcone^{\ast}}		
\newop{\ncone}{NC}		
\newop{\pcone}{PC}		
\newop{\hull}{\Delta}		

\newmacro{\cvx}{\mathcal{C}}		
\newmacro{\subd}{\partial}		

\newmacro{\minmax}{L}		

\newmacro{\minvar}{\theta}		
\newmacro{\minvaralt}{\alt\minvar}		
\newmacro{\minvars}{\Theta}		

\newmacro{\maxvar}{\phi}		
\newmacro{\maxvaralt}{\alt\maxvar}		
\newmacro{\maxvars}{\Phi}		

\newop{\Eucl}{\Pi}		
\newop{\logit}{\Lambda}		
\newop{\dkl}{KL}		

\newmacro{\hreg}{h}		
\newmacro{\breg}{D}		
\newmacro{\mprox}{P}		
\newmacro{\mirror}{Q}		
\newmacro{\fench}{F}		
\newmacro{\hstr}{K}		
\newmacro{\depth}{H}		
\newmacro{\proxdom}{\points^{\hreg}}		

\DeclarePairedDelimiterXPP{\proxof}[2]{\mprox_{#1}}{(}{)}{}{#2}		

\newmacro{\zone}{\mathbb{D}}		

\newmacro{\point}{x}		
\newmacro{\pointalt}{\alt\point}		
\newmacro{\pointaltalt}{\altalt\point}		
\newmacro{\points}{\mathcal{K}}		
\newmacro{\intpoints}{\relint\points}		

\newmacro{\base}{p}		
\newmacro{\basealt}{q}		
\newmacro{\basealtalt}{u}		

\newmacro{\open}{\mathcal{U}}		
\newmacro{\closed}{\mathcal{C}}		
\newmacro{\cpt}{\mathcal{K}}		
\newmacro{\nhd}{\mathcal{U}}		

\newop{\ex}{\mathbb{E}}		
\newop{\prob}{\mathbb{P}}		
\newop{\Var}{Var}		
\newop{\simplex}{\hull}		

\providecommand\given{}		

\DeclarePairedDelimiterXPP{\exof}[1]{\ex}{[}{]}{}{
\renewcommand\given{\nonscript\,\delimsize\vert\nonscript\,\mathopen{}} #1}

\DeclarePairedDelimiterXPP{\probof}[1]{\prob}{(}{)}{}{
\renewcommand\given{\nonscript\:\delimsize\vert\nonscript\:\mathopen{}} #1}

\DeclarePairedDelimiterXPP{\oneof}[1]{\one}{\{}{\}}{}{
\renewcommand\given{\nonscript\,\delimsize\vert\nonscript\,\mathopen{}} #1}

\newmacro{\sample}{\omega}		
\newmacro{\samples}{\Omega}		

\newmacro{\filter}{\mathcal{F}}		
\newmacro{\probspace}{(\samples,\filter,\prob)}		

\newmacro{\event}{E}       
\newmacro{\eventalt}{H}       
\newmacro{\mean}{\mu}		
\newmacro{\sdev}{\sigma}		
\newmacro{\variance}{\sdev^{2}}		

\newcommand{\est}[1]{\hat #1}		

\newmacro{\signal}{V}		
\newmacro{\step}{\gamma}		
\newmacro{\learn}{\eta}		
\newmacro{\mix}{\eps}		
\newmacro{\conf}{\delta}		
\newmacro{\mindiff}{d^{\ast}}		
\newmacro{\thres}{\ell^{\ast}}		

\newmacro{\proper}{\tau}		

\newmacro{\error}{Z}		
\newmacro{\noise}{U}		
\newmacro{\bias}{b}		
\newmacro{\brown}{W}		

\newmacro{\serror}{\theta}		
\newmacro{\snoise}{\xi}		
\newmacro{\sbias}{\psi}		

\newmacro{\sbound}{M}		
\newmacro{\bbound}{B}		
\newmacro{\noisepar}{\sdev}		
\newmacro{\noisevar}{\variance}		

\addauthor[Kyriakos]{KL}{Purple}

\addauthor[Nick]{NB}{Red}

\addauthor[Pan]{PM}{MediumBlue}

\newmacro{\round}{\mathcal{R}} 

\begin{document}


\title{
Learning in Games with Quantized Payoff Observations}

\author
[K.~Lotidis]
{Kyriakos Lotidis$^{\ddag}$}
\address{$^{\ddag}$\,%
Department of Management Science \& Engineering, Stanford University.}
\EMAIL{klotidis@stanford.edu}
\author
[P.~Mertikopoulos]
{Panayotis Mertikopoulos$^{\diamond,\star}$}
\address{$^{\diamond}$\,%
Univ. Grenoble Alpes, CNRS, Inria, Grenoble INP, LIG, 38000 Grenoble, France.}
\address{$^{\star}$\,%
Criteo AI Lab.}
\EMAIL{panayotis.mertikopoulos@imag.fr}
\author
[N.~Bambos]
{Nicholas Bambos$^{\dagger}$}
\address{$^{\dagger}$\,%
Department of Electrical Engineering and Management Science \& Engineering, Stanford University.}
\EMAIL{bambos@stanford.edu}


\maketitle




\newacro{LHS}{left-hand side}
\newacro{RHS}{right-hand side}
\newacro{iid}[i.i.d.]{independent and identically distributed}
\newacro{lsc}[l.s.c.]{lower semi-continuous}
\newacro{NE}{Nash equilibrium}
\newacroplural{NE}[NE]{Nash equilibria}

\newacro{FTRL}{``follow the regularized leader''}
\newacro{FTQL}{follow the quantized leader}
\newacro{IWE}{importance-weighted estimator}
\newacro{EW}{exponential \textpar{{\normalfont or} multiplicative} weights}

\begin{abstract}
%
%
This paper investigates the impact of feedback quantization on multi-agent learning.
In particular, we analyze the equilibrium convergence properties of the well-known \ac{FTRL} class of algorithms when players can only observe a quantized (and possibly noisy) version of their payoffs.
In this information-constrained setting, we show that coarser quantization triggers a qualitative shift in the convergence behavior of \ac{FTRL} schemes.
Specifically, if the quantization error lies below a threshold  value (which depends \emph{only} on the underlying game and not on the level of uncertainty entering the process or the specific \ac{FTRL} variant under study), then
\begin{enumerate*}
[(\itshape i\hspace*{1pt}\upshape)]
\item
\ac{FTRL} is attracted to the game's strict Nash equilibria with arbitrarily high probability;
and
\item
the algorithm's asymptotic rate of convergence remains the same as in the non-quantized case.
\end{enumerate*}
Otherwise, for larger quantization levels, these convergence properties are lost altogether:
players may fail to learn anything beyond their initial state, even with full information on their payoff vectors.
This is in contrast to the impact of quantization in continuous optimization problems, where the quality of the obtained solution degrades smoothly with the quantization level.
\end{abstract}

\allowdisplaybreaks		
\acresetall		


\section{Introduction}
\label{sec:introduction}

In the implementation of distributed learning and control systems, observations and feedback often need to be quantized down to the bit-resolution allowed by the sensing/sampling and data communication rates.
This is driven by various design pressures, including sensing/sampling and communication bandwidth constraints, as well as computation, memory, and power limitations.
In particular, such challenges are ubiquitous in current and emerging distributed systems (like the Internet of Things or edge/mobile computing), where edge devices must often contend with granular, reduced-precision data and measurements.
For example, a mobile device may only be able to measure the quality of its downlink channel up to a relatively low precision and then request setting the downlink transmitter power (which also affects other devices via interference) based on low-rate feedback.
Likewise, an edge computing node may only be able to receive a low-bit representation of the data of a control-plane application and must then process and resubmit this data using some low-bit encoding.

Reduced-precision settings of this type can be modeled efficiently by assuming that, in addition to any random factors affecting the process, observable quantities are also \emph{quantized}, reflecting the granularity of the measurement\,/\,communication process.
With this in mind, our paper examines \emph{quantized multi-agent learning processes} that unfold as follows:
\begin{enumerate}
\item
At each stage $\run=\running$, every participating agent selects an action from some finite set.
\item
Each agent's reward is determined by their chosen action and that of all other participating agents.
\item
Agents observe a noisy quantized version of their rewards, they update their actions, and the process repeats.
\end{enumerate}
In terms of the agents' learning dynamics \textendash\ \ie the way that they update their actions \textendash\ we consider the widely studied \acdef{FTRL} class of algorithms, as introduced by \cite{SSS06} and containing as special cases the seminal multiplicative/exponential weights algorithm of \cite{Vov90,LW94,ACBFS95},
as well as the standard projection dynamics of \cite{Fri91,Hop99a}.
Within this setting, we aim to address the following questions:
\begin{enumerate}
\item 
What is the \emph{impact of quantization} on the learning process relative to the non-quantized case?
\item
Is there \emph{robust deterioration} \textendash\ \ie graceful degradation, as opposed to abrupt collapse \textendash\ of the outcome of the learning process as the coarseness of the quantization increases?
\end{enumerate}

\para{Related work}

The literature on learning in games has traditionally focused on identifying when a learning process converges to equilibrium \textendash\ locally or globally.
In this regard, a widely known result is that the empirical frequency of play under no-regret learning converges to the game's set of \emph{coarse correlated equilibria}
\cite{MV78,HMC00}. 
However, since this set may contain highly undesirable, dominated strategies \cite{VZ13}, this convergence result typically needs to be refined.

On that account, a very large body of work has focused on the sharper question of convergence to a \acdef{NE}, \ie a state from which no player has an incentive to deviate unilaterally.
This question is much more difficult and only partial results are known:
as a representative (but otherwise incomplete) list of relevant results, \cite{Bra16,CMS10,LC03,LC05} established the convergence of an ``adjusted'' variant of \ac{FTRL} to approximate \aclp{NE} in potential, $2\times m$, and $2\times\dotsm\times2$ games.
This convergence was established under the assumption that players receive perfect realizations of their in-game payoffs \textendash\ \ie there are no observation or measurement errors, random or otherwise.
More recently, and under similar feedback assumptions, \cite{GVM21} showed that, in any generic game, strict \aclp{NE} \textendash\ \ie \aclp{NE} where each player has a unique best response \textendash\ are \emph{precisely} the states that are stable and attracting under the (unadjusted) dynamics of \ac{FTRL} in discrete time.

The algorithmic stability and convergence results discussed above were achieved via the use of an \acdef{IWE} which provides a counterfactual surrogate for the payoff that a player would have obtained from an action that they did not actually pick.
The key property of this estimator is that its bias can be balanced against its variance so as to yield progressively more accurate payoff predictions with only a mild deterioration in precision.
In turn, this ``asymptotic unbiasedness'' property plays a major role in the convergence results discussed above because it allows players to eventually gravitate towards actions that yield consistently better payoffs against the ``mean field'' of the other players' actions.

However, this crucial property is lost the moment quantization enters the picture:
the granularity of the players' payoff observations can \emph{never} become finer than the quantization gap of their feedback\,/\,measurement mechanism, so any learning process would not be able to resolve this gap either.
Indeed, \emph{any} payoff estimator must contend with a persistent bias that disallows the resolution of payoffs corresponding to nearby mixed strategies \textendash\ \eg playing $(1/3,1/3,1/3)$ versus $(1/3,1/3-\eps,1/3+\eps)$ in Rock-Paper-Scissors for sufficiently small $\eps$.
As a result, any learning process that relies on gradual changes in the players' mixed strategies \textendash\ like \ac{FTRL} and its variants \textendash\ would seem unable to make consistent progress towards a \acl{NE}, even if starting relatively close.

\para{Our contributions}

Our analysis paints a different account of the above.
First, if the quantization error does not exceed a certain threshold value, we show that \ac{FTRL} with quantized feedback \textendash\ dubbed ``\emph{\acl{FTQL}}\acused{FTQL}'' \textpar{\ac{FTQL}} for short \textendash\ continues to identify strict \aclp{NE} \emph{with perfect accuracy}, despite the \emph{persistent bias} induced by the quantization process.
More precisely, we show that strict \aclp{NE} are locally stable and attracting with arbitrarily high probability under \ac{FTQL}, just as in the case of \ac{FTRL} with \emph{perfect} payoff-based feedback.
Second, we derive a series of sharp convergence rate estimates for \ac{FTQL} which echo the convergence speed of \ac{FTRL} with \emph{non-quantized} feedback as derived recently in \cite{GVM21b}.
Specifically,
\emph{despite the quantization},
the convergence rate of \ac{FTQL} differs from its non-quantized variant only by a multiplicative constant, showing that the algorithm's asymptotic rate of convergence remains otherwise unimpeded by the coarseness of the quantization scheme (as long as this coarseness does not exceed the critical level beyond which learning is impossible).

Importantly, this quantization threshold depends \emph{only} on the underlying game and is otherwise independent of the level of uncertainty involved and/or the specific \ac{FTQL} variant in play.
Beyond this threshold, the learning landscape changes abruptly and dramatically.
In particular, for larger values of the quantization gap, the convergence properties of \ac{FTQL} are lost altogether:
players may fail to learn anything beyond their initial state, even with full information on their mixed payoff vectors (even in simple $2\times2$ common interest games that are otherwise easy to learn).

This behavior comes in stark contrast to the impact of quantization in continuous optimization where the quality of the obtained solution degrades gracefully with the quantization gap \cite{Nes04,NN94,JNT11}.
This suggests a fundamental shift in design principles when dealing with game-theoretic problems as above:
the robust deterioration observed in the discretization of continuous optimization problems is no longer present, and the quantization granularity has to be tuned judiciously as a function of the agents' interactions.

\section{Background and motivation}
\label{sec:preliminaries}

\subsection{Games in normal form}

Throughout this paper, we consider normal form games with a finite number of players and a finite number of actions per player.
More precisely, we posit that each player, indexed by $\play \in \players = \braces{1,\dots,\nPlayers}$, has a finite set of \emph{actions} \textendash\ or \emph{pure strategies} \textendash\ $\pure_{\play} \in \pures_{\play}$ and a \emph{payoff function} $\pay_{\play}: \pures \rightarrow \R$, where $\pures \defeq \prod_{\play \in \players}\pures_{\play}$ denotes the set of all possible action profiles $\pure = (\pure_{1},\dotsc,\pure_{\nPlayers})$.
Players can mix their strategies, \ie play a probability distribution $\strat_{\play} \in \strats_{\play}\defeq \simplex\parens{\pures_{\play}}$ over their pure strategies, and we write $\strat = \parens{\strat_\start,\dots, \strat_{\nPlayers}} \in \strats \defeq \prod_{\play \in \players}\strats_{\play}$ for the associated \emph{mixed strategy profile}.
For notational convenience, we will also write $\dspace_{\play}\defeq\R^{\pures_{\play}}$ and $\dspace \defeq \prod_{\play \in \players}\dspace_{\play}$, for the space of payoff vectors of player $\play\in\players$ and the ensemble thereof.

Given a mixed strategy profile $\strat\in\strats$, we will use the standard shorthand $\strat = (\strat_{\play};\strat_{-\play})$ to keep track of the mixed strategy profile $\strat_{-\play}$ of all players other than $\play$, and we further define
\begin{enumerate}
[(\itshape i\hspace*{1pt}\upshape)]
\item
The \emph{expected payoff} of player $\play$ under $\strat$:
\begin{equation}
\pay_{\play}\parens{\strat}
	= \pay_{\play}\parens{\strat_{\play}, \strat_{-\play}} 
\end{equation}
\item
The \emph{mixed payoff vector} of player $\play$ under $\strat$:
\begin{equation}
\payv_{\play}\parens{\strat}
	= \parens{\pay_{\play}\parens{\pure_{\play};\strat_{-\play}}}_{\pure_{\play}\in\pures_{\play}}
\end{equation}
\end{enumerate}
In words, $\payv_{\play}(\strat) \in \dspace_{\play}$ simply collects the expected payoffs $\payv_{\play\pure_{\play}}\parens{\strat} \defeq \pay_{\play}\parens{\pure_{\play};\strat_{-\play}}$, $\pure_{\play}\in\pures_{\play}$, that player $\play\in\players$ would have obtained by playing $\pure_{\play}\in\pures_{\play}$ against the mixed strategy profile $\strat_{-\play}$ of all other players.
Then, aggregating over all players $\play\in\players$, we will also write $\payv\parens{\strat} = \parens{\payv_\start\parens{\strat},\dots,\payv_\nPlayers\parens{\strat}} \in \dspace$ for the ensemble of the players' mixed payoff vectors.

In terms of solution concepts, we say that a strategy profile $\eq$ is a \acdef{NE} if no player has an incentive to unilaterally deviate from it, \ie
\begin{equation}
\label{def:Nash}
\tag{NE}
\pay_{\play}\parens{\eq}
	\geq \pay_{\play}\parens{\strat_{\play};\eq_{-\play}}
	\quad
	\forall \strat_{\play} \in \strats_{\play},
	\forall \play \in \players.
\end{equation}
Finally, we say that $\eq$ is a \emph{strict \acl{NE}} if the inequality in \eqref{def:Nash} is strict for all $\strat_{\play} \neq \eq_{\play}$, $\play\in\players$, \ie if any deviation from $\eq_{\play}$ results to a strictly worse payoff for the deviating player $\play \in \players$.
It is straightforward to verify that a strict equilibrium $\eq\in\strats$ is also \emph{pure} in the sense that each player assigns positive probability only to a single pure strategy. 

\subsection{Quantization: definitions and impact on learning}

\subsubsection{Basics of quantization}
As we discussed in the introduction,
our paper concerns models of repeated play where all observable quantities \textendash\ the players' payoffs, the associated vectors, etc. \textendash\ are subject to rounding and/or precision cutoffs.
To formalize this, let $\ell>0$ be the \emph{quantization error} of the players' observation\,/\,measurement device, and let $\round\from\R\to \ell\Z \equiv \{\dotsc,-2\ell,-\ell,0,\ell,2\ell,\dotsc\}$ be the associated \emph{quantization operator} which reduces to the identity on $\ell\Z$, and which maps any real number $\point\in\R$ to an integer multiple $\round(\point) \in \ell\Z$ of $\ell$ such that $\abs{\point - \round(\point)} \leq \ell/2$ for all $\point\in\R$.
For example,
the ``floor'' operation $\point \mapsto \floor{\point}$ has quantization error $\ell=2$ (since $\sup_{\point}\abs{\point - \floor{\point}} = 1$),
whereas
the ``round half away from zero'' (or ``commercial rounding'') operation $\round(\point) = \mathrm{sgn}(\point) \floor{\abs{\point} + 1/2}$ in Python and Java has a quantization error of $\ell = 1$.

Vectorizing this construction in the obvious way, we will write $\round(\vvec) \defeq (\round(\vvec_{k}))_{k=1,\dotsc,\vdim}$ for an arbitrary vector $\vvec\in\R^{\vdim}$.
Then, by construction, $\round$ reduces to the identity on $(\ell\Z)^{\vdim}$ and we have
\begin{equation}
\label{eq:quant-error}
\supnorm{\round(\vvec) - \vvec}
	\leq \ell/2
	\quad
	\text{for all $\vvec\in\R^{\vdim}$}.
\end{equation}
Unless explicitly mentioned otherwise, we will not assume a specific quantization operator in the sequel, and we will state our results only as a function of the quantization error $\ell$.

\subsubsection{The impact on learning}
To motivate the analysis to come, we provide below two examples where the process of quantization can lead to significant challenges in multi-agent learning, even in the simplest case where players observe their full (quantized) payoff vectors.

For concreteness, we will present our examples in the context of the well-known \acdef{EW} algorithm \cite{Vov90,LW94,ACBFS95} which, in our setting, can be written as
\begin{equation}
\label{eq:EW}
\tag{EW}
\begin{aligned}
\state_{\play\pure_{\play},\run}
	&\propto \exp(\dstate_{\play\pure_{\play},\run})
	\\
\dstate_{\play,\run+1}
	&= \dstate_{\play,\run} + \step_{\run} \signal_{\play,\run}
\end{aligned}
\end{equation}
where
\begin{enumerate*}
[(\itshape i\hspace*{1pt}\upshape)]
\item
$\state_{\play,\run} \in \strats_{\play}$ is the mixed strategy of player $\play\in\players$ at the $\run$-th stage of the process;
\item
$\signal_{\play,\run}$ is an approximation of the player's mixed payoff vector $\payv_{\play}(\state_{\run})$ which we discuss in detail below;
\item
$\dstate_{\play,\run} \in \dspace_{\play}$ is an auxiliary ``score vector'' that aggregates payoff information (so $\dstate_{\play\pure_{\play},\run}$ indicates the propensity of player $\play$ to employ the pure strategy $\pure_{\play}\in\pures_{\play}$);
and
\item
$\step_{\run} > 0$ is a ``learning rate'' (or step-size) parameter that controls the weight with which new information enters the algorithm.
\end{enumerate*}

With all this in hand,
the examples that follow are intended to highlight two critical issues:
\begin{enumerate*}
[\itshape a\upshape)]
\item
the evolution of \eqref{eq:EW} when $\signal_{\play,\run}$ is obtained by rounding $\payv_{\play}(\state_{\run})$ at different precision cutoffs;
and
\item
the difference between learning in $\fingame$ with quantized feedback versus learning with \emph{non-quantized} feedback in a quantized version $\round\parens{\fingame}$ of the original game.
\end{enumerate*}

\begin{example}[The role of the quantization error]
Consider a two-player common-interest game with $\pures_1 = \braces{a_1, a_2}, \pures_2 = \braces{b_1,b_2}$, and rewards given by the following payoff matrix:
\begin{center}
    \begin{tabularx}{0.45\textwidth} { 
  | >{\centering\arraybackslash}X 
  | >{\centering\arraybackslash}X 
  | >{\centering\arraybackslash}X | }
 \hline
  \emph{Player} $1 / 2$ & $b_1$ & $b_2$ \\
 \hline
 $a_1$ & $99.1$ & $100.9$  \\
 \hline
 $a_2$ & $100.9$  & $99.1$ \\
\hline
\end{tabularx}
\end{center}
Clearly, $(a_1,b_2)$ is a strict \acl{NE} of the game.

We now examine the case where, at each round $\run=\running$, both players observe their quantized mixed payoff vectors $\signal_{\play,\run} = \round\parens{{\payv\parens{\curr}}}$, $\play=1,2$, and subsequently update their strategies according to \eqref{eq:EW}.
The specific quantization schemes we consider are as follows:

\begin{enumerate}
[(\itshape i\hspace*{1pt}\upshape)]
\item
``\emph{Round to closest even away from zero}'' ($\ell = 2$):
this scheme maps $\point$ to the closest even integer and resolves ties by moving away from $0$, \ie $\round(\point) = 2\mathrm{sgn}(\point) \floor{\abs{\point}/2 + 1/2}$.
Then, for any initial mixed strategy profile $\state_{\start} \in \strats$ that assigns positive probability to all actions,
all coordinates of $\payv\parens{\start}$ will lie in the interval $\parens{99.1,100.9}$, so every entry of $\round\parens{\payv\parens{\start}}$ will in turn be equal to $100$.
We thus conclude that all coordinates of $\dstate_{\run}$ will be increased by the same amount in the iterative step $\dstate_{\run} \gets \dstate_{\run+1}$.
Since this constant increase disappears under the normalization step in \eqref{eq:EW}, we readily obtain $\curr = \init$ for all $\run=\running$,
\ie the players' strategy profile remains unchanged for all time in this learning model.%
\footnote{Note here that the precise values of the game are not important:
we would obtain the same result if we replaced $\braces{99.1,100.9}$ with $\braces{99+\eps, 101-\eps}$ for any $\eps>0$ sufficiently small.}

\item
``\emph{Round half away from zero}'' ($\ell = 1$):
as discussed above, this scheme maps $\point$ to the closest integer and resolves ties by moving away from $0$, \ie $\round(\point) = \mathrm{sgn}(\point) \floor{\abs{\point} + 1/2}$. For simplicity, we assume that the learning rate is constant, say $\gamma_\run = 1, \forall \run$.
Now, taking $(\state_{a_1,\start}, \state_{a_2,\start}) = (0.8, 0.2)$ and $(\state_{b_1,\start}, \state_{b_2,\start}) = (0.2, 0.8)$, we readily obtain $\round\parens{\payv_{a_1}\parens{\state_\start}} = \round\parens{\payv_{b_2}\parens{\state_\start}} = 101$, and $\round\parens{\payv_{a_2}\parens{\state_\start}} = \round\parens{\payv_{b_1}\parens{\state_\start}} = 99$.\\
As a result, the corresponding score differences for $\run = \start$ satisfy: $\dstate_{a_1,\run+1} - \dstate_{a_2,\run+1}
	> \dstate_{a_1,\run} - \dstate_{a_2,\run}$ and $\dstate_{b_2,\run+1} - \dstate_{b_1,\run+1}
	> \dstate_{b_2,\run} - \dstate_{b_1,\run}$
from which we readily get $\state_{a_1,\run+1} > \state_{a_1,\run}$ and $\state_{b_2,\run+1} > \state_{b_2,\run}$
Therefore, inductively we have that $\payv_{a_1}\parens{\curr}$ and $\payv_{b_2}\parens{\curr}$ increase as $\run$ grows, while $\payv_{a_2}\parens{\curr}$ and $\payv_{b_1}\parens{\curr}$ decrease.
We thus obtain $\round\parens{\payv_{a_1}\parens{\state_\run}} = \round\parens{\payv_{b_2}\parens{\state_\run}} = 101$, and $\round\parens{\payv_{a_2}\parens{\state_\run}} = \round\parens{\payv_{b_1}\parens{\state_\run}} = 99$
for all $\run$.
Hence:
\begin{equation}
\begin{aligned}
\dstate_{a_1,\run+1} - \dstate_{a_2,\run+1}
	&= \dstate_{a_1,\run} - \dstate_{a_2,\run} +2 \\
	& = \dstate_{a_1,\start} - \dstate_{a_2,\start} + 2\run
\end{aligned}
\end{equation}
from which we readily get
\begin{equation}
\frac{\exp(\dstate_{a_1,\run+1})}{\exp(\dstate_{a_2,\run+1})}
	 = \frac{\exp(\dstate_{a_1,\start})}{\exp(\dstate_{a_2,\start})} \exp(2\run)
\end{equation}
Taking $\run \rightarrow \infty$ we obtain $(\state_{a_1,\run},\state_{a_2,\run})\rightarrow (1,0)$ as $\run \rightarrow \infty$, and, likewise:
\(
(\state_{b_1,\run},\state_{b_2,\run})\rightarrow (0,1)
\)
as $\run \rightarrow \infty$.
We thus conclude that $\curr$ converges to a strict \acl{NE}.
\end{enumerate}

From the above, we see that for two different quantization lengths, the learning process may exhibit a completely different behavior:
in (i) it remains static throughout the execution of the algorithm,
whereas in (ii) $\state_{\run}$ converges to a strict \acl{NE} of the underlying game. As we will see later, there is a threshold value $\ell$ associated with the minimum payoff differences, where this transition is sharp.
\endenv
\end{example}

\begin{example}[Learning with quantized feedback vs. learning in the quantized game]
This example is intended to highlight the difference between learning in $\fingame$ with quantized feedback versus learning with perfect feedback in a quantized version $\round\parens{\fingame}$ of $\fingame$.
As before, suppose there are two players, $1$ and $2$, with action spaces $\pures_1 = \braces{a_1, a_2}$ and $\pures_2 = \braces{b_1,b_2}$ respectively, and let $\round(\point) = \mathrm{sgn}(\point) \cdot \floor{\abs{\point} + 1/2}$.
The payoff matrix of the original game $\fingame$, along with the quantized version of it, is shown below:

\begin{center}
    \begin{tabularx}{0.45\textwidth} { 
  | >{\centering\arraybackslash}X 
  | >{\centering\arraybackslash}X 
  | >{\centering\arraybackslash}X | }
 \hline
  \emph{Player} $1 \backslash 2$ & $b_1$ & $b_2$ \\
 \hline
 $a_1$ & $0.04 \xrightarrow{\round} 0$ & $0.8\xrightarrow{\round} 1$  \\
 \hline
 $a_2$ & $0.8\xrightarrow{\round} 1$  & $0.04\xrightarrow{\round} 0$ \\
\hline
\end{tabularx}
\end{center}

We denote by $\state_{\run}, \dstate_{\run}$ and $\Tilde{\state}_{\run}, \Tilde{\dstate}_{\run}$ the sequences of states generated by \eqref{eq:EW} on $\fingame$ and $\round\parens{\fingame}$ respectively.
Moreover, we assume for concreteness that $\gamma_\run = 1$ for all $\run$, and $\state_\start = \Tilde{\state}_\start$ with $(\state_{a_1,\start}, \state_{a_2,\start}) = (0.6, 0.4)$ and $(\state_{b_1,\start}, \state_{b_2,\start}) = (0.4, 0.6)$. So, the different procedures are as follows:
\begin{enumerate}
[(\itshape i\hspace*{1pt}\upshape)]
\item
``\emph{$\fingame$ with quantized feedback}'': In this setting, players observe $V_{\run} = \round\parens{{\payv\parens{\curr}}}$ at the $\run$-th stage. By the initial conditions, we obtain $V_\start = 0$,
which means that all coordinates of the score vector remain unchanged,
so, inductively, we get
$\dstate_\run = \dstate_\start$
and hence $\state_\run = \state_\start$ for all stages, \ie the learning process does not evolve. 
\item
``\emph{$\round\parens{\fingame}$ without quantization}'': Unlike the previous setting, the players observe the full payoff vector of $\round\parens{\fingame}$, \ie $\signal_{\run} = \mathbb{E}_{\tilde\pure_{\run} \sim \tilde\state_{\run}}\bracks{\round\parens{\payv\parens{\tilde\pure_{\run}}}}$. By the initial conditions, we have $(V_{a_1,\start}, V_{a_2,\start}) = (0.6, 0.4)$ and $(\signal_{b_1,\start}, \signal_{b_2,\start}) = (0.4, 0.6)$. Therefore, the corresponding score differences for $\run = \start$ satisfy: $\Tilde{\dstate}_{a_1,\run+1} - \Tilde{\dstate}_{a_2,\run+1}
	> \Tilde{\dstate}_{a_1,\run} - \Tilde{\dstate}_{a_2,\run}$ and $\Tilde{\dstate}_{b_2,\run+1} - \Tilde{\dstate}_{b_1,\run+1}
	> \Tilde{\dstate}_{b_2,\run} - \Tilde{\dstate}_{b_1,\run}$.
With a similar reasoning as in Example 1(ii), we have: $\state_{a_1,\run+1} > \state_{a_1,\run}$ and $\state_{b_2,\run+1} > \state_{b_2,\run}$.
Iterating over $n$, we get: $(\state_{a_1,\run},\state_{a_2,\run})\rightarrow (1,0)$ and $(\state_{b_1,\run},\state_{b_2,\run})\rightarrow (0,1)$
\ie $\curr$ converges to a strict equilibrium of $\round(\fingame)$.
\end{enumerate}

The above shows a remarkable difference in behavior:
in the case of $\fingame$ with quantized feedback,
players learn nothing beyond their initial state;
by contrast, learning with \emph{perfect} feedback in the quantized game $\round\parens{\fingame}$ converges to the strict \acl{NE} $(a_1,b_2)$.
This serves to highlight the fact that learning with quantized feedback cannot be compared to learning in a quantized game:
the players' end behavior is drastically different in the two cases.
\endenv
\end{example}

\section{The learning model}
\label{sec:model}

We now proceed to describe our general model for learning with quantized feedback;
for ease of reference, we will refer to this scheme as \acdef{FTQL}.

Viewed abstractly, our model is based on the standard \ac{FTRL} template \cite{SSS06} run with quantized (and possibly noisy) payoff observations as follows:
\begin{equation}
\label{eq:FTQL}
\tag{FTQL}
\begin{aligned}
\state_{\play,\run}
	&= \mirror_{\play}\parens{\dstate_{\play,\run}}
	\\
\dstate_{\play,\run+1}
	&= \dstate_{\play,\run} + \step_{\run} \signal_{\play,\run}
\end{aligned}
\end{equation}
In more detail,
the defining elements of \eqref{eq:FTQL} are
\begin{enumerate*}
[(\itshape i\hspace*{1pt}\upshape)]
\item
the approximate payoff vectors $\signal_{\play,\run} \in \dspace_{\play}$ which are reconstructed from the players' payoff observations;
and
\item
the players' ``choice maps'' $\mirror_{\play}\from\dspace_{\play}\to\strats_{\play}$ which determine each player's mixed strategy $\state_{\play,\run} \in \strats_{\play}$ as a function of the ``aggregate payoff'' variables $\dstate_{\play,\run} \in \dspace_{\play}$.
\end{enumerate*}
In the rest of this section, we describe both of these elements in detail;
for a pseudocode implementation of the method, see also \cref{alg:FTQL} below.

\subsubsection{The feedback process}

The vanilla version of \ac{FTRL} assumes that each player $\play\in\players$ observes the full (mixed) payoff vector $\signal_{\play,\run} \gets \payv_{\play}\parens{\curr}$ in order to update their individual score vector $\dstate_{\play,\run}$ at each stage $\run$.
However, in our model, we only assume that players observe a quantized \textendash\ and possibly noisy \textendash\ version of their in-game, realized payoffs.
Specifically, if $\est\pure_{\play,\run}\in\pures_{\play}$ denotes the action (pure strategy) chosen by the $\play$-th player at stage $\run$, we assume that each player receives as feedback the quantized reward
\begin{equation}
\label{eq:feedback}
\est\pay_{\play,\run}
	= \round\bracks{\pay_{\play}(\est\pure_{\play,\run};\est\pure_{-\play,\run}) + \snoise_{\play,\run}}
\end{equation}
where
$\round$ is a quantization operator with gap $\ell$ (\cf \cref{sec:preliminaries})
and
$\snoise_{\play,\run} \in \R$, $\run=\running$, is a random, zero-mean error capturing all sources of uncertainty in the process.
Specifically, letting $\filter_{\run}$ denote the history (natural filtration) of $\state_{\run}$, we will make the following statistical assumptions for $\snoise_{\run}$:
\begin{subequations}
\label{eq:noise}
\begin{align}
\label{eq:zeromean}
\text{Zero-mean:}
	&\;\;
\exof{\snoise_{\play,\run} \given \filter_{\run}}
	= 0
	&&
	\\
\label{eq:variance}
\text{Finite variance:}
	&\;\;
\exof{\abs{\snoise_{\play,\run}}^{2} \given \filter_{\run}}
	\leq \noisevar
	&&
\end{align}
\end{subequations}
\ie $\snoise_{\run}$ is an $L^{2}$-bounded martingale difference sequence relative to the history of play up to stage $\run$ (inclusive).

To reconstruct their payoff vectors from the quantized feedback model \eqref{eq:feedback}, we further assume that players employ the \acdef{IWE}
\begin{equation}
\label{eq:IWE}
\tag{IWE}
\signal_{\play\pure_{\play},\run}
    = \frac{\oneof{\est\pure_{\play,\run} = \pure_{\play}}}{\est\state_{\play\pure_{\play},\run}}
    \est\pay_{\play,\run}
\end{equation}
where
\(
\est{\state_{\play\pure_{\play},\run}}
    = (1-\mix_{\run})\state_{\play\pure_{\play},\run} + \mix_{\run}/ \abs{\pures_{\play}}
\)
denotes the probability with which player $\play$ selects action $\pure_{\play} \in \pures_{\play}$ at stage $\run$ given the mixed strategy profile $\state_{\play,\run} \in \strats_{\play}$ and an ``\emph{explicit exploration}'' parameter $\mix_{\run} > 0$.
The role of this parameter will be discussed in detail in the next section.

\subsubsection{The players' choice maps}

As mentioned above, the second defining element of \eqref{eq:FTQL} is the players' choice map $\mirror_{\play}\from\dspace_{\play} \to \strats_{\play}$ whose role is to translate the ``aggregate score'' vectors $\dstate_{\play,\run}\in\dspace_{\play}$ into mixed strategies $\state_{\play,\run} = \mirror_{\play}(\dstate_{\play,\run}) \in \dspace_{\play}$.
This choice map is in turn defined as a regularized best response of the form
\begin{equation}
\label{eq:mirror}
\mirror_{\play}(\dpoint_{\play})
	= \argmax\nolimits_{\strat_{\play} \in \strats_{\play}}
		\braces{\braket{\dpoint_{\play}}{\strat_{\play}} - \hreg_{\play}(\strat_{\play})}
\end{equation}
where $\hreg_{\play}\from\strats_{\play} \to \R$ denotes the method's namesake ``regularizer function''.

For concreteness, we will focus on a class of decomposable regularizers of the form $\hreg_{\play}(\strat_{\play}) = \sum_{\pure_{\play}\in\pures_{\play}} \theta_{\play}(\strat_{\play\pure_{\play}})$ where the ``kernel function'' $\theta_{\play}\from[0,1]\to\R$ is
\begin{enumerate*}
[(\itshape i\hspace*{1pt}\upshape)]
\item
continuous on $[0,1]$;
\item
twice differentiable on $(0,1]$;
and
\item
strongly convex, \ie $\inf_{t\in(0,1]} \theta_{\play}''(z) > 0$.
\end{enumerate*}
Two standard examples of such functions are:
\begin{enumerate}
[left=0pt]
\item
The \emph{entropic regularizer} $\theta_{\play}(z) = z\log z$:
a standard calculation shows that the induced choice map is $\mirror_{\play}(\dpoint_{\play}) = (\exp(\dpoint_{\play\pure_{\play}}))_{\pure_{\play}\in\pures_{\play}} \big/ \sum_{\pure_{\play}\in\pures_{\play}} \exp(\dpoint_{\play\pure_{\play}})$ which leads to the \acl{EW} update template \eqref{eq:EW} of \cref{sec:preliminaries}.

\item
The \emph{Euclidean regularizer} $\theta_{\play}(z) = z^{2}/2$:
trivially, the induced choice map is the closest point projection $\mirror_{\play}(\dpoint_{\play}) = \argmin_{\strat_{\play}\in\strats_{\play}} \twonorm{\dpoint_{\play} - \strat_{\play}}$, and the induced scheme is the projection dynamics \cite{Fri91,Hop99a}.
\end{enumerate}

An important distinction between these regularizers is that $\theta_{\play}'(0^{+}) = -\infty$ for the entropic regularizer while $\theta_{\play}'(0^{+})$ is finite for the Euclidean one.
Regularizers that have the former behavior are called \emph{steep} and have the property that the induced mirror map is interior-valued;
regularizers with the latter behavior are called \emph{non-steep} and have surjective mirror maps \cite{MS16}.
This behavior is captured by the \emph{rate function}
\begin{equation}
\label{eq:ratefn}
\phi_{\play}(\dpoint)
	= \begin{cases}
		0
			&\quad
			\text{if $\dpoint \leq \theta_{\play}'(0^{+})$}
		\\
		1
			&\quad
			\text{if $\dpoint \geq \theta_{\play}'(1^{-})$}
		\\
		(\theta_{\play}')^{-1}(\dpoint)
			&\quad
			\text{otherwise}
   \end{cases}
\end{equation}
As we shall see below, this rate function plays a crucial role in determining the rate of convergence of \eqref{eq:FTQL}.

\begin{algorithm}[tbp]
\begin{algorithmic}[1]
\State
	\textbf{Initialize:} $\dstate_\start$
\For{$\run = \start,\afterstart,\dots$}
\State
	$\state_{\play,\run} \leftarrow \mirror_{\play}\parens{\dstate_{\play,\run}}$
\State
	Update sampling strategy: $\;\est{\state_{\play\pure_{\play},\run}} \gets (1-\mix_{\run})\state_{\play\pure_{\play},\run} + \frac{\mix_{\run}}{\abs{\pures_{\play}}}$
\State
	Sample $\hat{\pure}_{\run} \sim \hat{\state}_{\run}$
\State
	Observe realized payoff: $\;\est\pay_{\play,\run} \gets \round \parens{\pay_{\play}\parens{\pure_{\play};\pure_{-\play}} + \snoise_{\play,\run}}$
\State
	Estimate payoff vector through \eqref{eq:IWE}: $$\signal_{\play\pure_{\play},\run} \gets {\oneof{\hat{\pure}_{\play,\run} = \pure_{\play}} \over \est{\state_{\play\pure_{\play},\run}}}\est\pay_{\play,\run}$$
\State
	Update score vectors: $\dstate_{\play,\run+1} \gets \dstate_{\play,\run}+\step_{\run} \signal_{\play,\run}$
\EndFor
\end{algorithmic}
\caption{\Acf{FTQL}}
\label{alg:FTQL}
\end{algorithm}

\section{Analysis and results}
\label{sec:results}

We are now in a position to proceed with the convergence analysis of the quantized learning scheme \eqref{eq:FTQL}.
The first thing to note is that a finite game may admit several Nash equilibria \textendash\ an odd number generically \textendash\ so it is not reasonable to expect a global convergence result that applies to all games.
For this reason, we will focus below on states that are \emph{locally} stable and attracting:

\begin{definition}
\label{def:stable}
Let $\state_{\run}$, $\run=\running$, be the sequence of mixed strategy profiles generated by \eqref{eq:FTQL}.
We then say that $\eq\in\strats$ is:
\begin{enumerate}
[left=.5em,itemsep=0pt]

\item
\emph{Stochastically stable} if, for every confidence level $\conf>0$ and every neighborhood $\nhd$ of $\eq$ in $\strats$, there exists a neighborhood $\nhd_{\start}$ of $\eq$ in $\strats$ such that
\begin{equation}
\label{eq:stable}
\probof{\state_{\run}\in\nhd \; \text{for all $\run$} \given \state_{\start} \in \nhd_{\start}}
	\geq 1-\conf.
\end{equation}

\item
\emph{Attracting} if, for every confidence level $\conf>0$, there exists a neighborhood $\nhd_{\start}$ of $\eq$ in $\strats$ such that
\begin{equation}
\label{eq:attract}
\probof{\state_{\run}\to\eq \; \text{as $\run\to\infty$} \given \state_{\start} \in \nhd_{\start}}
	\geq 1-\conf.
\end{equation}

\item
\emph{Stochastically asymptotically stable} if it is stochastically stable and attracting.
\end{enumerate}
\end{definition}

Informally, the above states that $\eq$ is stochastically stable if every trajectory $\state_{\run}$ of \eqref{eq:FTQL} that starts sufficiently close to $\eq$ remains nearby with arbitrarily high probability;
in addition, if $\state_{\run}$ converges to $\eq$ as well, then $\eq$ is stochastically \emph{asymptotically} stable \cite{Kha12,Rob12}.
On that account, states that are (stochastically) asymptotically stable under \eqref{eq:FTQL} are the only states that can be considered as viable, stable outcomes of the learning process.

In the context of \ac{FTRL} with perfect, \emph{non-quantized} payoff observations, it is known that a state is stochastically asymptotically stable \emph{if and only if} it is a \emph{strict} \aclp{NE} of $\fingame$ \cite{GVM21}.
With this in mind, and given that the advent of quantization can only worsen the attraction properties of any given point (\cf the relevant discussion in \cref{sec:preliminaries}), we will exclusively focus below on the asymptotic stability and attraction properties of strict \aclp{NE} under \eqref{eq:FTQL}.

In this regard, our main result can be summarized along the following two axes:
\begin{enumerate}
\item
If the quantization error $\ell$ is smaller than a threshold value $\thres$ that depends only on the underlying game, every strict \acl{NE} of $\fingame$ is stochastically asymptotically stable under \eqref{eq:FTQL}.
\item
Conditioned on the above, convergence to a strict equilibrium $\eq\in\strats$ occurs at a rate of $\onenorm{\state_{\run} - \eq}
	\leq \phi\parens*{-\Theta\parens*{\insum_{\runalt=\start}^{\run} \step_{\runalt}}}$, where $\phi$ is the rate function \eqref{eq:ratefn}.
\end{enumerate}

The idea of our proof is to find a set of suitable initial conditions for the quantized version of $\payv\parens{\curr}$ to remain in the interior of the \textit{normal cone} $\ncone\parens{\eq}$ of $\strats$ at $\eq$ throughout the execution of the algorithm. For this, we need to delve into the geometry of $\ncone\parens{\eq}$ and find the limitations in the quantization legnth $\ell$ that guarantee that $\curr$ will be contained in the desired region.

We start with the following lemma that gives a specific description of the \textit{normal cone} at a vertex $\eq$ of the polytope $\strats$.

\begin{restatable}{lemma}{cone}
\label{lem:cone}
Let $\eq$ be of the form $\parens{e_{1\pureq_{1}},\dots,e_{\nPlayers\pureq_{\nPlayers}}}$, where $e_{\play\pureq_{\play}} \in \R^{\abs{\pures_{\play}}}$ a standard basis vector. Then the normal cone of $\strats$ at $\eq$ can be expressed as: 
\begin{equation}
\ncone\parens{\eq} = \braces{w \in \dspace: w_{\play\pure_{\play}} - w_{\play\pureq_{\play}} \leq 0, \forall \play \in \players, \pure_{\play} \in \pures_{\play}}
\end{equation}
\end{restatable}

\begin{Proof}
    We have that $\strats = \braces{\strat \in \R^{\abs{\pures}} : \sum_{\pure_{\play} \in \pures_{\play}}\strat_{\play\pure_{\play}} = 1, \strat_{\play\pure_{\play}} \geq 0, \forall \pure_{\play} \in \pures_{\play}, \play \in \players}$, for $\abs{\pures} = \abs{\pures_1}+\dots+\abs{\pures_\nPlayers}$. We can equivalently write it in standard form, as:
\begin{equation}
\strats
	= \setdef
	{\strat \in \R^{\abs{\pures}}}
	{C\strat = e, \strat_{\play\pure_{\play}} \geq 0, \forall \pure_{\play} \in \pures_{\play}, \play \in \players}
\end{equation}
where $C$ is a $\nPlayers \times \parens{\abs{\pures_1}+\dots+\abs{\pures_\nPlayers}}$ matrix whose $i$-th row is $c_{\play}^T = \parens{0,\dots,0,1,\dots,1,0,\dots,0}$, with ones in positions $\parens{\abs{\pures_1}+\dots +\abs{\pures_{\play-1}}+1},\dots, \parens{\abs{\pures_1}+\dots +\abs{\pures_{\play}}}$. Then, every vertex $\overline{\strat}$ of $\strats$ is of the form: $\overline{\strat}_{\play\overline{\pure}_{\play}} = 1$ for some $\overline{\pure}_{\play} \in \pures_{\play}$ and $\overline{\strat}_{\play\pure_{\play}} = 0, \forall \pure_{\play} \neq \overline{\pure}_{\play} \in \pures_{\play}, \forall \play \in \players$. Hence, $\eq$ is an extreme point of the bounded polytope $\strats$ and the set of adjacent vertices of $\eq$ is the set $\mathcal{Z} = \braces{\eq - e_{\play\pureq_{\play}}+e_{\play\pure_{\play}}: \pure_{\play} \in \pures_{\play}, \play \in \players}$. 
Now, let $\mathcal{C} = \braces{w: \inner{w}{z - \eq} \leq 0, \forall z \in \mathcal{Z}}$. 
The tangent cone of $\strats$ at $\eq$ equals to the closure of the cone of feasible directions at $\eq$, and since $\strats$ is a convex polytope, we get:
\begin{equation}
	\tcone\parens{\eq} = \text{cone}\parens{\braces{z - \eq: z \in \mathcal{Z}}}
\end{equation}
Since $\ncone\parens{\eq} = \parens{\tcone\parens{\eq}}^{\circ}:= \braces{w: \inner{w}{\strat} \leq 0, \forall \strat \in \tcone\parens{\eq}}$, it remains to show that $\mathcal{C} = \ncone\parens{\eq}$. Clearly, $\ncone\parens{\eq} \subset \mathcal{C}$, since $z - \eq \in \tcone\parens{\eq}, \forall z \in \mathcal{Z}$. For the opposite direction, take $w \in \mathcal{C}$. Then, for $\strat \in \tcone\parens{\eq}$, i.e. $\strat = \sum_{j = 1}^\runalt \lambda_j\parens{z_j - \eq}$ for $z_j \in \mathcal{Z}, \lambda_j \geq 0$, we have $\inner{w}{\strat} \leq 0$, since $\inner{w}{z_j - \eq} \leq 0, \forall z_j \in \mathcal{Z}$. Therefore, we get $w \in \ncone\parens{\eq} \implies \mathcal{C} \subset \ncone\parens{\eq}$, and the result follows.
\end{Proof}

Given this representation of the normal cone at the vertices of $\strats$, we can derive several geometric properties of strict \aclp{NE}. Informally, the next lemma states that the payoff vector $\payv\parens{\eq}$ at a strict equilibrium $\eq$  belongs to the interior of $\ncone\parens{\eq}$, and also gives the distance from the cone's boundary.

\begin{restatable}{lemma}{polar}
\label{lem:polar}
Let $\eq = \parens{\pureq_{1},\dots,\pureq_{\nPlayers}} \in \strats$ be a strict \acl{NE} and
let $\mindiff$ be defined as per \eqref{eq:mindiff}.
\begin{enumerate}[(a)]
    \item
    If $\ell \leq \mindiff$, then $\ball \parens{\payv\parens{\eq}, {\ell \over 2}} \subseteq \ncone\parens{\eq}$, where $\ball$ is with respect to $\supnorm{\cdot}$
    \item
    If $\ell \leq {\mindiff \over m}$ for $m \in \N$ and $d = \mindiff-m \ell$, then for any $w \in \ball \parens{\payv\parens{\eq}, {d \over 2}}$, we have: $w_{\play \pure_{\play}} - w_{\play \pureq_{\play}} + m \ell \leq 0$, for any $\pure_{\play} \in \pures_{\play}, \play \in \players$.
\end{enumerate}
\end{restatable}

\begin{Proof}
\textit{(a)} Let $w \in \ball \parens{\payv\parens{\eq}, {\ell \over 2}}$. We have: $\abs{w_{\play\pure_{\play}} - \payv_{\play\pure_{\play}}\parens{\eq}} \leq {\ell \over 2}, \forall \pure_{\play} \in \pures_{\play}, \play \in \players$
Then, for any $z_{\play\pure_{\play}} \in \mathcal{Z}$, we have:
\begin{align*}
w_{\play\pure_{\play}} - w_{\play\pureq_{\play}} & \leq \payv_{\play\pure_{\play}}\parens{\eq} - \payv_{\play\pureq_{\play}}\parens{\eq} + \ell \notag\\
& = \pay_{\play}\parens{\pure_{\play};\pureq_{-\play}} - \pay_{\play}\parens{\pureq_{\play};\pureq_{-\play}} +\ell \notag\\
&\leq 0 
\end{align*}
from which we conclude that $w \in \ncone\parens{\eq}$. \\[2mm]
\textit{(b)} Let $w \in \ball \parens{\payv\parens{\eq}, {\ell \over 2}}$ and $\Tilde{w}$ defined as follows:
\begin{equation}
\Tilde{w}_{\play\pure_{\play}} = 
	\begin{cases}
	w_{\play\pure_{\play}} - {m\ell \over 2}
		&\quad
		\text{if $\pure_{\play} = \pureq_{\play}$}
	\\
	w_{\play\pure_{\play}} + {m\ell \over 2}
		&\quad
		\text{otherwise}
	\end{cases}
\end{equation}
Then, $\supnorm{\Tilde{w}-w} \leq {m\ell \over 2}$, and hence we have:
\begin{equation}
\supnorm{\Tilde{w}-\payv\parens{\eq}} \leq \supnorm{\Tilde{w}-w}+ \supnorm{w-\payv\parens{\eq}} \leq {\mindiff \over 2}
\end{equation}
from which we get that $\Tilde{w} \in \ball \parens{\payv\parens{\eq}, {\mindiff \over 2}}$, \ie $\Tilde{w} \in \ncone\parens{\eq}$ due to part \textit{(a)}. Therefore, we conclude that for any $\play \in \players, \pure_{\play} \in \pures_{\play}$: 
\begin{equation}
\Tilde{w}_{\play\pure_{\play}} - \Tilde{w}_{\play\pureq_{\play}} \leq 0 \implies w_{\play \pure_{\play}} - w_{\play \pureq_{\play}} + m \ell \leq 0
\end{equation}
\end{Proof}

Now, we are ready to state and prove our main theorem.
\begin{restatable}{theorem}{main}
\label{thm:main}
Let $\eq = (\pureq_{1},\dotsc,\pureq_{\nPlayers})$ be a strict \acl{NE} of $\fingame$ and let
\begin{equation}
\label{eq:mindiff}
\mindiff
	= \min\nolimits_{\play\in\players}
	\min\nolimits_{\pure_{\play}\in\pures_{\play}\setminus\{\pureq_{\play}\}}
		\braces{\pay_{\play}\parens{\pureq_{\play};\pureq_{-\play}} - \pay_{\play}\parens{\pure_{\play};\pureq_{-\play}}}
\end{equation}
denote the minimum payoff difference incurred by a unilateral off-equilibrium deviation.
Assume further that \eqref{eq:FTQL} is run with quantization error $\ell < \mindiff/3$ and step-size and exploration parameters such that
\begin{equation}
\label{eq:params}
\sum_{\run=\start}^{\infty} \step_{\run}
	= \infty,
	\;\;
\sum_{\run=\start}^{\infty} \step_{\run} \mix_{\run}
	< \infty
	\;\;
	\text{and}
	\;\;
\sum_{\run=\start}^{\infty} \frac{\step_{\run}^{2}}{\mix_{\run}^{2}}
	< \infty.
\end{equation}
Then $\eq$ is stochastically asymptotically stable and, for all trajectories converging to $\eq$, we have
\begin{equation}
\label{eq:rate}
\onenorm{\state_{\play,\run} - \eq_{\play}}
	\leq 2 \sum_{\pure_{\play} \neq \pureq_{\play}} \phi_{\play}\parens*{-c \tau_{\run} +o\parens{\tau_{\run}}}
\end{equation}
where $\tau_{\run} = \sum_{\runalt=\start}^{\run} \step_{\runalt}$ and $\const \in (0,\mindiff - 3\ell)$.
\end{restatable}

\begin{Proof}
Since \eqref{eq:FTQL} updates the score vector $\dstate_\run$ at each stage $\run$, we need a connection between the variables in the dual space, $\dspace$, and the ones in the primal, $\strats$. This connection is conveniently expressed through the so-called \textit{score-dominant sets} \cite{GVM21}. Formally, \cite{GVM21} shows that  for any $\eps >0$, there exist $M_{\play,\eps}$ for all $\play \in \players$ so that
\begin{equation}
\label{eq:score-subset}
\prod_{\play \in \players} Q_{\play}\parens{\mathcal{W}_{\play}\parens{M_{\play,\eps}}} \subseteq \mathcal{U}_\eps
\end{equation}
where: 
\begin{equation}
\mathcal{W}_{\play}\parens{M_{\play,\eps}} = \braces{\dstate_{\play}: \dstate_{\play,\pureq_{\play}}-\dstate_{\play,\pure_{\play}}>M_{\play,\eps},  \forall \pure_{\play} \neq \pureq_{\play}}
\end{equation} and
\begin{equation}
\label{eq:neighborhood}
\nhd_\eps = \braces{\strat \in \strats: \strat_{\play\pure_i^*} > 1-\eps, \forall \play \in \players}.
\end{equation}
Therefore, our goal in the sequel will be to find a set of initial conditions so that the corresponding score differences $\dstate_{\play\pureq_{\play}, \run}-\dstate_{\play\pure_{\play},\run}$, stay large enough throughout the stages of the algorithm for all players $\play \in \players$. We will now proceed to prove each part of the theorem separately.

\para{Stochastic stability}
To begin with, fix a confidence level $\conf >0$, and let $\nhd$ be a neighborhood of $\eq$ in $\strats$. Invoking \cref{lem:polar} with $m = 3$, we see that any $w \in \ball \parens{\payv\parens{\eq}, {d \over 2}}$ satisfies $w_{\play \pure_{\play}} - w_{\play \pureq_{\play}} + 3 \ell < 0$. By continuity of $\payv$, there exist a neighborhood $\overline{\nhd}$ of $\eq$ and $c >0$ such that $\overline{\nhd} \subseteq \nhd$ and $\payv_{\play \pure_{\play}}(\strat) - \payv_{\play \pureq_{\play}}(\strat) + 3 \ell \leq -c$, for $\strat \in \overline{\nhd}$.
Then, by \eqref{eq:score-subset},\eqref{eq:neighborhood}, there exist $\eps_0 > 0, M_{\play,\eps_0}$, for all $\play \in \players$ such that:
\begin{enumerate}[(a)]
\item $\nhd_{\eps_0} \subseteq \overline{\nhd} \subseteq \nhd$
\item $\prod_{\play \in \players} Q_{\play}\parens{\mathcal{W}_{\play}\parens{M_{\play,\eps_0}}} \subseteq \nhd_{\eps_0}$
\end{enumerate}


For our analysis, we decompose the approximate payoff vector $\signal_{\run}$ in components as follows
\begin{equation}
\signal_{\play\pure_{\play},\run}
	= \exof{\round\parens{\payv_{\play\pure_{\play}} \parens{\pure_{\run}} + \snoise_{\play,\run}} \given \curr[\filter]}
	+ \noise_{\play\pure_{\play},\run}
	+ \bias_{\play\pure_{\play},\run}
\end{equation}
where
$\payv_{\play\pure_{\play}}(\pure_{\run}) = \pay_{\play}(\pure_{\play};\pure_{-\play,\run})$ and:
\begin{enumerate}
\item
$\noise_{\play\pure_{\play},\run} \defeq \signal_{\play\pure_{\play},\run} - \exof{\round\parens{\payv_{\play\pure_{\play}}\parens{\est\pure_{\run}} + \snoise_{\play,\run}} \given \curr[\filter]}$ is a zero-mean error process.
\item
$\bias_{\play\pure_{\play},\run} \defeq \exof{\round\parens{\payv_{\play\pure_{\play}}\parens{\hat{\pure}_{\run}} + \snoise_{\play,\run}} \given \curr[\filter]} - \exof{\round\parens{\payv_{\play\pure_{\play}}\parens{\pure_{\run}} + \snoise_{\play,\run}} \given \curr[\filter]}$ is a systematic (non-zero-mean) error process due to (a) quantization; and (b) sampling from $\hat{\state}_{\run}$ instead of $\curr$.
\end{enumerate}
It is important to highlight that previous techniques of \cite{GVM21b} and references therein can no longer be applied to this setting, because the bias term $\bias_{\play,\run}$ is \emph{not} diminishing, but \emph{persistent} in all stages of the \eqref{eq:FTQL} due to the quantization error.
By comparison, all previous analyses require that any bias entering a learning algorithm vanish appropriately in the long run.

To proceed, we denote by $\Tilde{V}_{\play,\run}:= \exof{\round\parens{\payv_{\play}\parens{\pure_{\run}} + \snoise_{\play,\run}\cdot e} \given \curr[\filter]}$ where $e$ is a vector of ones of appropriate dimension, and by $\Tilde{\bias}_{\play\pure_{\play},\run} := \payv_{\play\pure_{\play}}\parens{\hat{\state}_{\run}} - \payv_{\play\pure_{\play}}\parens{{\curr}}$. 
\begin{restatable}{claim}{ineqs}
\label{claim:ineqs}
The following inequalities hold: $\supnorm{\Tilde{V}_{\play,\run} - \payv_{\play}\parens{\state_{\run}}} \leq {\ell \over 2}$ and
$\supnorm{\bias_{\play,\run} - \Tilde{\bias}_{\play,\run}} \leq \ell$
\end{restatable}
\begin{restatable}{claim}{ineps}
\label{claim:eps}
$\exof{\dnorm{\Tilde{\bias}_{\run}} \given \curr[\filter]} = O\parens{\eps_{\run}}$ and $\exof{\dnorm{U_{\run}}^2 \given \curr[\filter]} = O\parens{1 / \eps_{\run}^2}$
\end{restatable}

\cref{claim:ineqs} follows from \eqref{eq:quant-error} and some algebraic derivations, while \cref{claim:eps} holds due to Lipschitz continuity of $\payv(\cdot)$, compactness of $\strats$ and $\mathcal{L}^2$-boundedness of $\snoise$. The proofs are omitted due to lack of space.

With these two claims in hand, we will focus on player $\play$ and drop the index $i$ altogether.
For any $\pure\neq\pure^* \in \pures$ and assuming $\state_{\runalt} \in \nhd_{\eps_0}$ for $k = \start,\dots, \run$ we have:
\begin{align}
	\dstate_{\pure,\run+1} - \dstate_{\pure^*,\run+1} & = \dstate_{\pure,\run} - \dstate_{\pure^*,\run} + \gamma_{\run}\parens{V_{\pure,\run} - V_{\pure^*,\run}}
	\notag\\
	& = \dstate_{\pure,\run} - \dstate_{\pure^*,\run} +
	\gamma_{\run}\bracks{\parens{\Tilde{V}_{\pure,\run} - \Tilde{V}_{\pure^*,\run}}+ \parens{\bias_{\pure,\run} - \bias_{\pure^*,\run}} 
	 + \parens{\noise_{\pure,\run} - \noise_{\pure^*,\run}}}
	 \notag\\
	& \leq \dstate_{\pure,\run} - \dstate_{\pure^*,\run} +\gamma_{\run}[\parens{\payv_\pure\parens{\curr} - \payv_{\pure^*}\parens{\curr} + \ell} +\parens{\Tilde{\bias}_{\pure,\run}-\Tilde{\bias}_{\pure^*,\run} + 2\ell} \notag\\
	 &\quad +\parens{\noise_{\pure,\run} - \noise_{\pure^*,\run}}]
	 \notag\\
	& = \dstate_{\pure,\run} - \dstate_{\pure^*,\run} +\gamma_{\run}\parens{\payv_\pure\parens{\curr} - \payv_{\pure^*}\parens{\curr} + 3\ell}  +\gamma_{\run}\parens{\Tilde{\bias}_{\pure,\run}-\Tilde{\bias}_{\pure^*,\run}} \notag\\ &\quad+\gamma_{\run}\parens{\noise_{\pure,\run} - \noise_{\pure^*,\run}}
	\notag\\
	& = \dstate_{\pure,\start} - \dstate_{\pure^*,\start} +\sum_{\runalt = \start}^{\run}\gamma_{\runalt}\parens{\payv_\pure\parens{\state_{\runalt}} - \payv_{\pure^*}\parens{\state_{\runalt}} + 3\ell}  +\sum_{\runalt = \start}^{\run}\gamma_{\runalt}\parens{\Tilde{\bias}_{\pure,\runalt}-\Tilde{\bias}_{\pure^*,\runalt}} \notag\\
    &\quad +\sum_{\runalt = \start}^{\run}\gamma_{\runalt}\parens{\noise_{\pure,\runalt} - \noise_{\pure^*,\runalt}}
\end{align}
where the inequality step follows from \cref{claim:ineqs}.

We will first bound the term $\sum_{\runalt = \start}^\run\gamma_{\runalt}\parens{\noise_{\pure,\runalt} - \noise_{\pure^*,\runalt}}$ for all $\run \in \N$. If we define 
\begin{equation}
	R_{\run}:=\sum_{\runalt = \start}^\run\gamma_{\runalt}\parens{\noise_{\pure,\runalt} - \noise_{\pure^*,\runalt}}
\end{equation}
it is easy to see that $R_n$ is a martingale, as $\exof{\abs{R_{\run}} \given \curr[\filter]} < \infty$ and $\exof{R_{\run+1} \given \curr[\filter]} = R_{\run}, \forall \run$.

Moreover, for $K_1 > 0$, whose value will be determined later, we define the sequence of events
\begin{equation}
	D_{\run,K_1} = \braces{\sup_{\runalt \leq \run} \abs{\noise_{\pure,\runalt} - \noise_{\pure^*,\runalt}}\geq K_1}
\end{equation}
and 
\begin{equation}
	D_{K_1} = \braces{\sup_{\runalt \geq \start} \abs{\noise_{\pure,\runalt} - \noise_{\pure^*,\runalt}}\geq K_1}.
\end{equation}

By Doob's maximal inequality (Theorem 2.4, \cite{HH80}), we have 
\begin{equation}
\label{eq:doob_noise}
\probof{D_{\run,K_1}} \leq {\exof{R_{\run}^2} \over K_1^2}.
\end{equation}
Furthermore, we have that
\begin{equation}
    \exof{R_{\run}^2} = \sum_{\runalt = \start}^{\run}\gamma_{\runalt}^2\exof{\parens{\noise_{\pure,\runalt} - \noise_{\pure^*,\runalt}}^2}
\end{equation}
since 
\begin{equation}
	\exof{\noise_{\pure_1,\runalt}\noise_{\pure_2,\runaltalt}} = \exof{\exof{\noise_{\pure_1,\runalt}\noise_{\pure_2,\runaltalt} \given \filter_{\runalt}}} = \exof{\noise_{\pure_1,\runalt}\exof{\noise_{\pure_2,\runaltalt} \given \filter_{\runalt}}} = 0
\end{equation}
as $\noise_{\pure_1,\runalt}$ is $\filter_{\runalt}$-measurable for $\runalt < \runaltalt$, $\pure_1,\pure_2 \in \braces{\pure,\pure^*}$ and $\exof{\noise_{\pure_2,\runaltalt} \given \filter_{\runalt}}=0$.

Moreover
\begin{equation}
	\exof{\parens{\noise_{\pure,\runalt} - \noise_{\pure^*,\runalt}}^2} \leq 2\exof{\dnorm{\noise_{\runalt}}^2} = 2\exof{\exof{\dnorm{\noise_{\runalt}}^2 \given \filter_{\runalt}}} = O(1/\eps^{2}_{\runalt})
\end{equation}
by \cref{claim:eps}, and, therefore, there exists a constant $C_{1}$ such that
\begin{equation}
	\exof{R_{\run}^2} \leq C_{1} \sum_{\runalt = 0}^{\run} \frac{\gamma^{2}_{\runalt}}{\eps^{2}_{\runalt}}
\end{equation}
Hence, \eqref{eq:doob_noise} becomes:
\begin{equation}
\label{eq:noise_bound}
\probof{D_{\run,K_1}} \leq {\exof{R_{\run}^2} \over K_1^2} \leq \frac{C_{1} \sum_{\runalt = \start}^\run \gamma_{\runalt}^2 / \eps^{2}_{\runalt}}{K_{1}^{2}}
\end{equation}
Taking $\run \to \infty$ we get:
\begin{equation}
\probof{D_{K_1}} \leq \frac{C_{1} \sum_{\runalt = \start}^{\infty} \gamma_{\runalt}^2 / \eps^{2}_{\runalt}}{K_{1}^{2}} < \infty
\end{equation}
as $\braces{D_{\run,K_1}}_{\run \in \N}$ is an increasing sequence of events converging to $D_{K_1}$. Finally, setting the value of $K_1$ equal to:
\begin{equation}
	K_1 = \sqrt{\frac{2C_{1} \sum_{\runalt = \start}^{\infty} \gamma_{\runalt}^2 / \eps^{2}_{\runalt}}{\conf}}
\end{equation}
we get:
\begin{equation}
	\probof{D_{K_1}} \leq {\conf \over 2}
\end{equation}
Now, for the term $\sum_{\runalt = 0}^\run \gamma_{\runalt}\parens{\Tilde{\bias}_{\pure,\runalt} - \Tilde{\bias}_{\pure^*,\runalt}}$, we have that:
\begin{equation}
\abs*{\sum_{\runalt = \start}^\run \gamma_{\runalt}\parens{\Tilde{\bias}_{\pure,\runalt} - \Tilde{\bias}_{\pure^*,\runalt}}} \leq \sum_{\runalt = \start}^\run \gamma_{\runalt}\abs{\Tilde{\bias}_{\pure,\runalt} - \Tilde{\bias}_{\pure^*,\runalt}} \leq 2\sum_{\runalt = \start}^\run \gamma_{\runalt} \dnorm{\Tilde{\bias}_{\runalt}}
\end{equation}
As before, defining 
\begin{equation}
S_{\run}:=2\sum_{\runalt = \start}^\run \gamma_{\runalt} \dnorm{\Tilde{\bias}_{\runalt}}	
\end{equation}
it is easy to see that $S_{\run}$ is a submartingale, as $\exof{\abs{S_{\run}}} < \infty$ and $\exof{S_{\run+1} \given \curr[\filter]} \geq S_{\run}, \forall \run$.

Moreover, for $K_2 > 0$, whose value will be determined later, we define the sequence of events
\begin{equation}
	E_{\run,K_2} = \braces{\sup_{\runalt \leq \run} S_{\runalt} \geq K_2}
\end{equation}
and
\begin{equation}
	E_{K_2} = \braces{\sup_{\runalt \geq \start} S_{\runalt} \geq K_2}
\end{equation}
Again, by Doob's maximal inequality (Theorem 2.4, \cite{HH80}), we have 
\begin{equation}
	\probof{E_{\run,K_2}} \leq {\exof{S_{\run}} \over K_2} = 
	{2\sum_{\runalt = \start}^{\run} \gamma_{\runalt} \exof{\exof{\dnorm{\Tilde{\bias}_{\runalt}} \given \filter_{\runalt}}} \over K_2} \leq {C_{2} \sum_{\runalt = \start}^{\run} \gamma_{\runalt} \eps_{\runalt} \over K_2} < \infty
\end{equation}
for some constant $C_{2}$, by \cref{claim:eps}. Taking $\run \to \infty$ we get:
$$
\probof{E_{K_2}} \leq {C_{2} \sum_{\runalt = 0}^{\infty} \gamma_{\runalt} \eps_{\runalt} \over K_2} < \infty
$$
as $\braces{E_{\run,K_2}}_{\run \in \N}$ is an increasing sequence of events converging to $E_{K_2}$. Finally, setting the value of $K_2$ equal to:
\begin{equation}
	K_2 = \frac{2 C_{2}\sum_{\runalt = 0}^\infty \gamma_{\runalt} \eps_{\runalt}}{\conf}
\end{equation}
we get:
\begin{equation}
	\probof{E_{K_2}} \leq {\conf \over 2}
\end{equation}
Hence, by the union bound, we get: $\probof{D_{K_1}\cup E_{K_2}} \leq \conf$. Setting $M > M_{\eps_0} + K_1 + K_2$, we get that if $\dstate_\start \in \mathcal{W}\parens{M}$, i.e. $\dstate_{\pure,\start} - \dstate_{\pure^*,\start} < -M$ then:
\begin{align}
 \dstate_{\pure,\run+1} - \dstate_{\pure^*,\run+1}  &= \dstate_{\pure,\start} - \dstate_{\pure^*,\start} + \sum_{\runalt = \start}^\run \gamma_{\runalt}\parens{\payv_\pure\parens{\state_{\runalt}} - \payv_{\pure^*}\parens{\state_{\runalt}} + 3\ell} + \sum_{\runalt = \start}^\run\gamma_{\runalt}\parens{\noise_{\pure,\runalt} - \noise_{\pure^*,\runalt}} \notag\\
 & \quad + \sum_{\runalt = \start}^\run \gamma_{\runalt}\parens{\Tilde{\bias}_{\pure,\runalt} - \Tilde{\bias}_{\pure^*,\runalt}}
	\notag\\
	&\leq -M + K_1 + K_2 < -M_{\eps_0}
\end{align}
on the event $\parens{D_{K_1}\cup E_{K_2}}^c$, from which we get that $\next \in \nhd_{\eps_0}$, i.e. $\next \in \nhd$ with probability at least $1-\conf$. Therefore, we conclude that $\eq$ is stochastically stable.

\smallskip
\para{Stochastic asymptotic stability}
From the previous analysis, on the event $\parens{D_{K_1}\cup E_{K_2}}^c$ we have that:
\begin{equation}
\label{eq:rates}
	\dstate_{\pure,\run+1} - \dstate_{\pure^*,\run+1} < -M_{\eps_0} -c \sum_{\runalt = \start}^\run \gamma_{\runalt}
\end{equation}
Sending $n\to \infty$, we get that $\dstate_{\pure,\run+1} - \dstate_{\pure^*,\run+1} \rightarrow - \infty$, from which we have that for all $\Tilde{M} > 0, \dstate_{\runalt} \in \mathcal{W}\parens{\Tilde{M}}$ eventually. Hence, for all $\Tilde{\eps} >0, \state_{\runalt} \in \nhd_{\Tilde{\eps}}$ eventually, from which we get that $\state_{\runalt} \rightarrow \eq$ as $\runalt \rightarrow \infty$.


\para{Rates of convergence}
Finally, to establish the rate of convergence of \eqref{eq:FTQL}, let $\sum_{\runalt = \start}^{\run}\gamma_{\runalt}$ as $\tau_{\run}$ for all $\run$.
Since $R_\run = \sum_{\runalt = \start}^\run\gamma_{\runalt}\parens{\noise_{\pure,\runalt} - \noise_{\pure^*,\runalt}}$ is a martingale, $\lim_{\run \rightarrow \infty}\tau_\run = \infty$ and $\sum_{\runalt = 1}^{\infty}\tau^{-2}_{\runalt}\exof{\abs{\gamma_{\runalt}\parens{\noise_{\pure,\runalt} - \noise_{\pure^*,\runalt}}}^2 \given \filter_{\runalt}} \leq C \sum_{\runalt = 1}^{\infty}\tau^{-2}_{\runalt}\gamma^{2}_{\runalt}/\eps^{2}_{\runalt} < \infty$, for some constant $C>0$, by Strong law of large numbers for martingales (Theorem 2.18, \cite{HH80}), we have that
\begin{equation}
	\frac{R_{\run}}{\tau_{\run}} \rightarrow 0 \quad \text{a.s.}
\end{equation}
as $\run \rightarrow \infty$, i.e. $\probof{\Omega_1} = 1$ for $\Omega_1 = \braces{\frac{R_{\run}}{\tau_{\run}} \rightarrow 0, \; \text{as} \; \run \rightarrow \infty}$

Moreover, since $S_{\run}$ is a nonnegative submartingale with $\exof{S_\run}$ bounded for all $\run$, by Doob's submartingale convergence theorem (Theorem 2.5, \cite{HH80})
we obtain that there exist a random variable $S_{\infty}$ with $\exof{\abs{S_{\infty}}} < \infty$ and
\begin{equation}
	S_{\run} \rightarrow S_{\infty} \quad \text{a.s.}
\end{equation}
as $\run \rightarrow \infty$. Letting $A = \braces{S_{\run} \rightarrow S_{\infty}, \; \text{as} \; \run \rightarrow \infty}$, it holds $\probof{A} = 1$.

Since $S_{\infty} \in \mathcal{L}^{1}\parens{\prob}$, we get that $S_{\infty} < \infty$ a.s., which means that, if we define the set $B = \braces{S_{\infty} < \infty}$, we have $\probof{B} = 1$. Letting $\Omega_2 = A \cap B$, we have that $\probof{\Omega_2} = 1$, since 
\begin{equation}
	\probof{\Omega_{2}^c} = \probof{A^c \cup B^c} \leq \probof{A^c}+\probof{B^c} = 0 \Rightarrow \probof{\Omega_{2}^c} = 0	
\end{equation} 
Hence, on $\Omega_2$, we get that
\begin{equation}
	\frac{S_{\run}}{\tau_{\run}} \rightarrow 0 
\end{equation}
as $\run \rightarrow \infty$, since $S_\run \rightarrow S_{\infty} < \infty$ and $\tau_\run \rightarrow \infty$.

Therefore, since $\probof{\Omega_1} = \probof{\Omega_2} = 1$, we get that $\probof{\Omega_1\cap \Omega_2} = 1$, with the same reasoning as before. Now, let $\Omega_3$ be the event defined as $\Omega_3 = \braces{\curr \rightarrow \eq, \; \text{as} \; \run \rightarrow \infty}$.

Since $\curr \rightarrow \eq$ as $\run \rightarrow \infty$ on $\Omega_1\cap \Omega_2\cap \Omega_3$, we have that $\curr \in \overline{\nhd}$ eventually, i.e. there exists $\run_{0} \in \N$ such that $\curr \in \overline{\nhd}$ for all $\run \geq \run_{0}$, from which we get $\payv_{\play \pure_{\play}}(\strat) - \payv_{\play \pureq_{\play}}(\strat) + 3 \ell \leq -c$. Hence, for $\run \geq \run_{0}$, we obtain
\begin{align}
	\dstate_{\pure,\run+1} - \dstate_{\pure^*,\run+1} & = \dstate_{\pure,\run_{0}} - \dstate_{\pure^*,\run_{0}} +\sum_{\runalt = \run_{0}}^{\run}\gamma_{\runalt}\bracks{\parens{\Tilde{V}_{\pure,\runalt} - \Tilde{V}_{\pure^*,\runalt}}+ \parens{\bias_{\pure,\runalt} - \bias_{\pure^*,\runalt}} 
	 + \parens{\noise_{\pure,\runalt} - \noise_{\pure^*,\runalt}}} \notag\\
	& \leq \dstate_{\pure,\run_{0}} - \dstate_{\pure^*,\run_{0}} +\sum_{\runalt = \run_{0}}^{\run}\gamma_{\runalt}\bracks{\parens{\payv_\pure\parens{\state_{\runalt}} - \payv_{\pure^*}\parens{\state_{\runalt}} + 3\ell}  + \parens{\Tilde{\bias}_{\pure,\runalt}-\Tilde{\bias}_{\pure^*,\runalt}}\notag\\
    & \quad+ \parens{\noise_{\pure,\runalt} - \noise_{\pure^*,\runalt}}} \notag\\
	& \leq \dstate_{\pure,\run_{0}} - \dstate_{\pure^*,\run_{0}} - c \sum_{\runalt = \run_{0}}^{\run}\gamma_{\runalt} + \parens{S_{\run} - S_{\run_{0}-1}} +\parens{R_{\run} - R_{\run_{0}-1}} \notag\\
	& = - c \tau_{\run} + o\parens{\tau_{\run}}
\end{align}

Therefore, we have
\begin{align}
\label{eq:theta-primes}
	\theta^\prime\parens{\state_{\pure,\run+1}} & \leq \theta^\prime\parens{\state_{\pure^*,\run+1}} -c \tau_{\run} +o\parens{\tau_{\run}}
	\leq \theta^\prime\parens{1} -c \tau_{\run} +o\parens{\tau_{\run}}
\end{align}
from which we conclude
\begin{equation}
	\state_{\pure,\run+1} \leq \phi\parens{-c \tau_{\run} +o\parens{\tau_{\run}}}
\end{equation}
Aggregating over all strategies $\pure \in \pures, \pure \neq \pure^*$ we have:
\begin{equation}
\label{eq:rates}
	\onenorm{\eq - \state_{\run+1}} \leq 2 \sum_{\pure \neq \pure^*} \phi\parens{-c \tau_{\run} +o\parens{\tau_{\run}}}
\end{equation}
Finally, since $\probof{\Omega_1\cap \Omega_2} = 1$, we have that $\probof{\parens{\Omega_1\cap \Omega_2}\cap \Omega_3} = \probof{\Omega_3}$, and therefore, the convergence result holds for (almost) all trajectories converging to $\eq$.
\end{Proof}

The qualitative behavior of the dynamics as a function of the quantization length is shown in Fig. \ref{fig:simulations}. 

\begin{figure}
     \centering
     \begin{subfigure}[b]{\columnwidth}
        \begin{subfigure}[b]{0.30\columnwidth}
         \centering
         \includegraphics[scale=0.25]{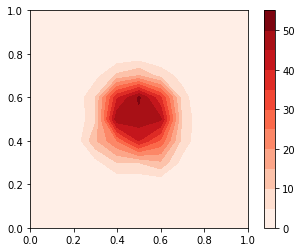}
         \label{fig:y equals x}
     \end{subfigure}
     \hfill
     \begin{subfigure}[b]{0.30\columnwidth}
         \centering
         \includegraphics[scale=0.25]{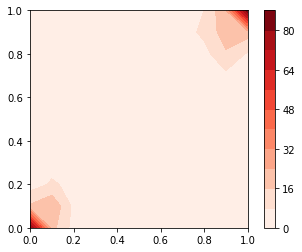}
         \label{fig:three sin x}
     \end{subfigure}
     \hfill
     \begin{subfigure}[b]{0.30\columnwidth}
         \centering
         \includegraphics[scale=0.25]{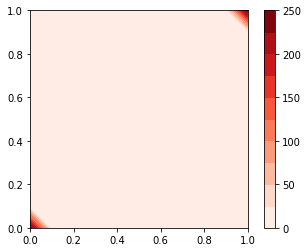}
         \label{fig:five over x}
     \end{subfigure}
     \subcaption{$\ell = 0$}
     \vspace*{2mm}
     \end{subfigure}
     \begin{subfigure}[b]{\columnwidth}
        \begin{subfigure}[b]{0.30\columnwidth}
         \centering
         \includegraphics[scale=0.25]{figs/zero_it0.png}
         \label{fig:y equals x}
     \end{subfigure}
     \hfill
     \begin{subfigure}[b]{0.30\columnwidth}
         \centering
         \includegraphics[scale=0.25]{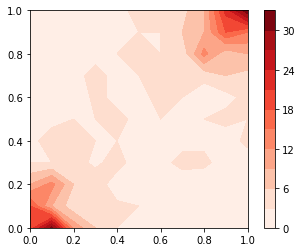}
         \label{fig:three sin x}
     \end{subfigure}
     \hfill
     \begin{subfigure}[b]{0.30\columnwidth}
         \centering
         \includegraphics[scale=0.25]{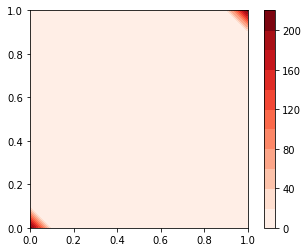}
         \label{fig:five over x}
     \end{subfigure}
     \subcaption{$\ell = 1.5$}
     \vspace*{2mm}
     \end{subfigure}
     \begin{subfigure}[b]{\columnwidth}
        \begin{subfigure}[b]{0.30\columnwidth}
         \centering
         \includegraphics[scale=0.25]{figs/zero_it0.png}
         \label{fig:y equals x}
     \end{subfigure}
     \hfill
     \begin{subfigure}[b]{0.30\columnwidth}
         \centering
         \includegraphics[scale=0.25]{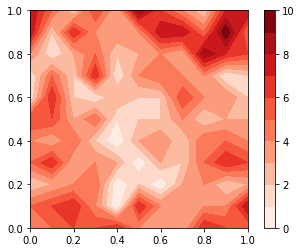}
         \label{fig:three sin x}
     \end{subfigure}
     \hfill
     \begin{subfigure}[b]{0.30\columnwidth}
         \centering
         \includegraphics[scale=0.25]{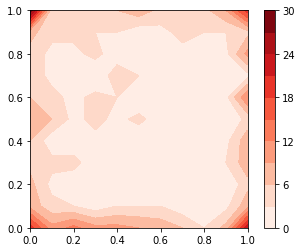}
         \label{fig:five over x}
     \end{subfigure}
     \subcaption{$\ell = 4$}
     \end{subfigure}
        \caption{
        Density `heat-maps' of the temporal evolution of \eqref{eq:FTQL} at times 0 (left), 50 (middle), 2000 (right) in 500 instances (sampled as $Y_1 \sim U[0,1]^{4}$) on a $2\times2$ symmetric game with $\pay(a_1, b _1) = \pay(a_2,b_2) = 5.1$, $\pay(a_1, b _2) = \pay(a_2,b_1) = 2.4$, $p = 0.75$, $q = 0.25$, and $\snoise_\run \sim U[-0.1,0.1]$ for three quantization errors: $\ell = 0,\,1.5,\,4$. Both $(a_1, b _1)$ and $(a_2, b _2)$ are pure Nash equilibria. In each plot, the horizontal axis corresponds to $x_{1a_1}$ and the vertical one to $x_{2b_1}$. We observe that for small quantization errors $(\ell = 0,\,1.5)$, the instances converge to the \aclp{NE} $(a_1, b _1)$ and $(a_2, b _2)$; however, for $\ell = 1.5$, the convergence is slower. On the contrary, for large quantization length $(\ell = 4)$, the behavior of the system is unpredictable and all the pure strategy profiles become attractors. Note the {\em gradual ``disintegration'' of convergence} with larger quantization error.
        } 
        \label{fig:simulations}
\end{figure}
Some corollaries and remarks are in order. We begin by discussing the possible schedules for the algorithm's step-size and exploration parameters:
here, a standard choice is to take $\step_{\run} \propto 1/\run^{\pexp}$ and $\mix_{\run} \propto 1/\run^{\qexp}$, with $\pexp,\qexp \geq 0$ chosen so as to satisfy \eqref{eq:params}.
A straightforward verification gives the conditions
\begin{equation}
\label{eq:params-schedule}
\pexp
	\leq 1,
	\quad
\pexp + \qexp
	> 1
	\quad
	\text{and}
	\quad
2\pexp - 2\qexp
	> 1
\end{equation}
which, in turn, provide the following explicit guarantee:

\begin{restatable}{corollary}{rateschedule}
\label{cor:rate-schedule}
Suppose that \eqref{eq:FTQL} is run with assumptions as in \cref{thm:main} and with step-size and exploration parameters satisfying \eqref{eq:params-schedule}.
Then, for all $\pexp>3/4$, \eqref{eq:FTQL} achieves
\begin{equation}
\label{eq:rate-schedule}
\onenorm{\state_{\play,\run} - \eq_{\play}}
	\leq 2 \insum_{\pure_{\play} \neq \pureq_{\play}} \phi_{\play}\parens*{-\Theta(\run^{1-\pexp})}.
\end{equation}
\end{restatable}

\begin{Proof}
From \eqref{eq:params-schedule} it is easy to see that $p \in (3/4,1]$. For $p=1$, $\tau_{\run} = \Theta\parens{\log\run}$, so \eqref{eq:FTQL} achieves faster convergence if $p<1$. 
Since for $p \in (3/4,1)$ we have that $\tau_{\run} = \Theta\parens{\run^{1-p}}$, the result is immediate from \eqref{eq:rates}.
\end{Proof}

\begin{remark*}
In the absence of quantization, \ac{FTRL} achieves a convergence rate of $\phi(-\sum_{\runalt=\start}^{\run}\step_{\runalt})$ \cite{GVM21b} with $p\in[0,1]$.
In our case, if the assumption for $\snoise$ is strengthened to almost sure boundedness (or sub-Gaussian increments), we can likewise relax the step-size requirements and achieve the rate \eqref{eq:rate-schedule} for \emph{any} $\pexp\in[0,1]$.
\endenv
\end{remark*}

Our next result concerns the rate of convergence of \eqref{eq:FTQL} for different choices of the mirror map $\mirror$ as defined in \eqref{eq:mirror}:

\begin{restatable}{corollary}{ratemirror}
\label{cor:rate-mirror}
Suppose that \eqref{eq:FTQL} is run with assumptions as in \cref{cor:rate-schedule}.
Then:
\begin{enumerate}
[left=0.5em,itemsep=0pt]

\item
The \acl{EW} variant of the algorithm \textpar{$\theta(z) = z\log z$} achieves convergence to strict \aclp{NE} at a rate of:  \begin{equation}
    \onenorm{\state_{\play,\run} - \eq_{\play}}
	\leq 2 \insum_{\pure_{\play} \neq \pureq_{\play}} \exp\parens*{-\Theta(\run^{1-\pexp})}.
\end{equation}
\item
The Euclidean variant of the algorithm \textpar{$\theta(z) = z^{2}/2$} achieves convergence to strict \aclp{NE} at a \textbf{finite} number of iterations.
\end{enumerate}
\end{restatable}

\begin{Proof}
\textit{1)} Since $\theta\parens{x} = x \log x$, we have $\phi\parens{x} = \exp(x-1)$. So, we obtain $\phi_{\play} \parens{-\Theta(\run^{1-\pexp})} = \exp\parens{-\Theta(\run^{1-\pexp})}$. Then, the result is immediate.\\[2mm]
\textit{2)} Since $\state_{\play,\run} \geq 0$, $\theta'$ increasing, \eqref{eq:theta-primes} becomes 
\begin{equation}
    \theta'(0) \leq \theta'(\state_{\pure,\run+1}) \leq \theta^\prime\parens{1} -c \tau_{\run} +o\parens{\tau_{\run}}
\end{equation}
and, since $\lim_{\run}\tau_{\run} = \infty$, we obtain 
\begin{equation}
    \tau_{\run} \geq \frac{1}{c}\parens{\theta^\prime\parens{1} - \theta'\parens{0} +o\parens{\tau_{\run}}}
\end{equation}
for large $\run$. Combining the above inequalities, we get
\begin{equation}
    \theta'(0) \leq \theta'(\state_{\pure,\run+1}) \leq \theta'(0)
\end{equation}
from which we conclude that $\state_{\pure,\run+1} = 0$ for sufficiently large $\run$, as per our claim.
\end{Proof}
The sharp separation between exponential and Euclidean variants of \ac{FTRL} with \emph{non-quantized} feedback was also observed in \cite{GVM21b}.
What is rather surprising here is that \Cref{cor:rate-mirror} echoes the non-quantized rates despite the presence of quantization:
the examples discussed in \cref{sec:preliminaries} show that the transition from convergence to non-convergence is sharp, so this would in turn suggest that the precise identification of a strict \acl{NE} in a finite number of iterations is not possible.
As we explain (\cf \cref{lem:cone,lem:polar}), what enables this result is the structure of the polar cone to $\strats$ at a strict equilibrium, which is reflected on the dependence of the rates on the difference $\mindiff - 3\ell$ via the constant $\const$:
when $\ell$ exceeds a critical value, any drift towards the strict equilibrium in question disappears, and \eqref{eq:FTQL} abruptly loses all its stability and convergence properties. 

For completeness, our last corollary concerns the rate of convergence of the players' actual sampling strategies $\est\state_{\run}$.
The precise result here is as follows:

\begin{restatable}{corollary}{ratesimple}
\label{cor:rate-sample}
With assumptions as in \cref{cor:rate-schedule}, the players' sampling strategies under \eqref{eq:FTQL} enjoy the rate:
\begin{equation}
\label{eq:rate-sample}
\onenorm{\est\state_{\play,\run} - \eq_{\play}}
	\leq 2 \sum_{\pure_{\play}\neq\pureq_{\play}} \phi_{\play}\parens{-\Theta(\run^{1-\pexp})}
		+ \Theta(1/\run^{\qexp})
\end{equation}
Accordingly, under the exponential variant of \eqref{eq:FTQL}, $\est\state_{\run}$ converges at a rate of $\Theta(1/\run^{\qexp})$ for all $\qexp\in(0,1/2)$.
\end{restatable}

\begin{proof}
By the triangle inequality, we obtain
\begin{align}
	\onenorm{\hat{\state}_{\play,\run} - \eq_{\play}} & \leq \onenorm{\hat{\state}_{\play,\run} - \state_{\play,\run}} + \onenorm{\state_{\play,\run} - \eq_{\play}}
	\notag\\
	&\leq  2 \sum_{\pure_{\play}\neq\pureq_{\play}} \phi_{\play}\parens{-c \tau_{\run} +o\parens{\tau_{\run}}} + \abs{\pures}\;\eps_{\run}
	\notag\\
	&= 2 \sum_{\pure_{\play}\neq\pureq_{\play}} \phi_{\play}\parens{-\Theta(\run^{1-\pexp})}
		+ \Theta(1/\run^{\qexp})
\end{align}
Finally, for the exponential variant of \eqref{eq:FTQL}, the rate is $\Theta(1/\run^{\qexp})$, since $ \phi_{\play}\parens{-\Theta(\run^{1-\pexp})} = \exp\parens{-\Theta(\run^{1-\pexp})}$.
\end{proof}
We conclude with the proof of the two intermediate claims used in the proof of \cref{thm:main}, namely:

\ineqs*

\ineps*

\begin{proof}[Proof of \cref{claim:ineqs}]
For any player $\play \in \players$ and any $\pure_\play \in \pures_\play$, we have
\begin{equation}
\payv_{\play\pure_{\play}} \parens{\pure_{\run}} + \snoise_{\play,\run} - {\ell \over 2} \leq \round\parens{\payv_{\play\pure_{\play}} \parens{\pure_{\run}} + \snoise_{\play,\run}} \leq  \payv_{\play\pure_{\play}} \parens{\pure_{\run}} + \snoise_{\play,\run} + {\ell \over 2}
\end{equation}

Taking $\exof{\; \cdot\given \curr[\filter]}$, we get:

\begin{equation}
\payv_{\play\pure_{\play}}\parens{\state_{\run}} - {\ell \over 2} \leq \exof{\round\parens{\payv_{\play\pure_{\play}} \parens{\pure_{\run}} + \snoise_{\play,\run}} \given \curr[\filter]} \leq  \payv_{\play\pure_{\play}}\parens{\state_{\run}} + {\ell \over 2}
\end{equation}

from which we get 
\begin{equation}
	\abs{\exof{\round\parens{\payv_{\play\pure_{\play}} \parens{\pure_{\run}} + \snoise_{\play,\run}} \given \curr[\filter]} - \payv_{\play\pure_{\play}}\parens{\state_{\run}}} \leq {\ell \over 2}
\end{equation}
This proves the first part of the claim.

For the second part, following the same procedure as before, we we have:
\begin{equation}
	\abs{\exof{\round\parens{\payv_{\play\pure_{\play}} \parens{\hat{\pure}_{\run}} + \snoise_{\play,\run}} \given \curr[\filter]} - \payv_{\play\pure_{\play}}\parens{\hat{\state}_{\run}}} \leq {\ell \over 2}
\end{equation}

Then, the second part of the lemma follows by triangle inequality.
\end{proof}

\begin{proof}[Proof of \cref{claim:eps}]
Since $\payv\parens{\cdot}$ is continuous on $\strats$ and $\strats$ is compact, $\payv\parens{\cdot}$ is also Lipschitz continuous. Hence, we have:
\begin{equation}
	\dnorm{\Tilde{\bias}_{\run}} = \dnorm{\payv\parens{\curr} - \payv\parens{\hat{\state}_{\run}}} \leq K \norm{\curr - \hat{\state}_{\run}} = O\parens{\eps_{\run}}
\end{equation}
from which the result follows.

For the other part, we have
\begin{align}
	\dnorm{U_{\play,\run}} \leq \dnorm{V_{\play,\run}} + \dnorm{\exof{\round\parens{\pay_{\play}\parens{\hat{\pure}_{\run}} + \snoise_{\play,\run}\cdot e} \given \curr[\filter]}}
\end{align}
where $e$ is the all ones vector. For the second term of the RHS, we have by triangle inequality:
\begin{equation}
\dnorm{\exof{\round\parens{\payv_{\play}\parens{\hat{\pure}_{\run}} + \snoise_{\play,\run}\cdot e} \given \curr[\filter]}} \leq \dnorm{\exof{\round\parens{\payv_{\play}\parens{\hat{\pure}_{\run}} + \snoise_{\play,\run}\cdot e} \given \curr[\filter]} -\payv_{\play}\parens{\hat{\state}_{\run}}} 
+ \dnorm{\payv_{\play}\parens{\hat{\state}_{\run}}} 
\end{equation}
By \cref{claim:ineqs}, we have:
\begin{equation}
  \supnorm{\exof{\round\parens{\payv_{\play}\parens{\hat{\pure}_{\run}} + \snoise_{\play,\run}\cdot e} \given \curr[\filter]} - \payv_{\play}\parens{\hat{\state}_{\run}}} \leq {\ell \over 2}  
\end{equation}
and therefore, $\dnorm{\exof{\round\parens{\payv_{\play}\parens{\hat{\pure}_{\run}} + \snoise_{\play,\run}\cdot e} \given \curr[\filter]} - \payv_{\play}\parens{\hat{\state}_{\run}}}$ is bounded. Moreover, we have:
\begin{equation}
	\dnorm{\payv_{\play}\parens{\hat{\state}_{\run}}} \leq \max_{\pure \in \pures}\abs{\pay_{\play}\parens{\pure}}
\end{equation}
Hence, combining the above, we get that the term $\dnorm{\exof{\round\parens{\payv_{\play}\parens{\hat{\pure}_{\run}} + \snoise_{\play,\run}\cdot e} \given \curr[\filter]}}$ is bounded.\\[2mm]
Regarding the term $V_{\play,\run}$, we have that each one if its entry $V_{\play\pure_{\play},\run}$, satisfies:
\begin{align}
	\abs{V_{\play\pure_{\play},\run}} & \leq {1 \over \hat{\state}_{\play\pure_{\play},\run}}\abs{\round\parens{\pay_{\play}\parens{\pure_{\play},\hat{\pure}_{-\play,\run}} + \snoise_{\play,\run}}}
	\notag\\
	& \leq {1 \over \hat{\state}_{\play\pure_{\play},\run}}\parens{{\ell \over 2} + \abs{\pay_{\play}\parens{\pure_{\play},\hat{\pure}_{-\play,\run}} + \snoise_{\play,\run}}}
	\notag\\
	& \leq {1 \over \hat{\state}_{\play\pure_{\play},\run}}\parens{{\ell \over 2} + \abs{\pay_{\play}\parens{\pure_{\play},\hat{\pure}_{-\play,\run}}} + \abs{\snoise_{\play,\run}}}
	\notag\\
	& \leq {\abs{\pures_{\play}} \over \eps_{\run}}\parens{{\ell \over 2} + \abs{\snoise_{\play,\run}} + \max_{\pure \in \pures} \abs{\pay_{\play}\parens{\pure}}}
\end{align}
Finally, since $\snoise_{\play,\run} \in \mathcal{L}^2\parens{\prob}, \forall \play$ and $\mathcal{L}^2\parens{\prob} \subset \mathcal{L}^1\parens{\prob}$, we get for any $\eps_{\run} >0$ that $\abs{V_{\play\pure_{\play},\run}}^2 \in \mathcal{L}^2\parens{\prob}$, from which we conclude that $\exof{\dnorm{V_{\run}}^2 \given \curr[\filter]} = O\parens{1 / \eps_{\run}^2}$. Then, the result follows.
\end{proof}
\section{Concluding remarks}
\label{sec:conclusion}

Our results show that the impact of quantization on learning in games is somewhat different than what one would perhaps expect:
instead of a gradual deterioration of the quality of learning as the quantization error increases, we see that \eqref{eq:FTQL} continues to identify strict \aclp{NE} \emph{perfectly} if the quantization is not too coarse, and the rate of convergence is as in the non-quantized case.
We find this property particularly appealing, as it shows that, if the feedback process is quantized judiciously, we can achieve significant gains in terms of memory storage and bandwidth expenditures \emph{without} compromising the quality of learning.\\

\bibliographystyle{IEEEtran}
\bibliography{bibtex/IEEEabrv.bib,bibtex/Bibliography.bib}

\end{document}